\documentclass[a4paper,USenglish,cleveref,thm-restate]{lipics-v2021}

\usepackage{amsmath,amsthm}
\usepackage{amssymb}
\usepackage{xspace}
\usepackage{soul}
\usepackage{thmtools}
\usepackage{bm}

\usepackage{graphicx}
\usepackage{rotating}

\usepackage{algorithm}
\usepackage{algpseudocode}

\renewcommand{\paragraph}[1]{\smallskip\noindent \textbf{#1} \;}

\DeclareMathOperator{\Forall}{\forall\,}
\DeclareMathOperator{\Exists}{\exists\,}
\DeclareMathOperator{\Nexists}{\nexists\,}

\newcommand{\ttrue}{\ensuremath{\mathtt{true}}\xspace}

\newtheorem{assumption}{Assumption}

\usepackage{tikz}

\usetikzlibrary{math,arrows.meta,decorations.markings}

\newcommand{\Objects}{\ensuremath{\mathit{Objects}}\xspace}
\newcommand{\Processes}{\ensuremath{\mathit{Processes}}\xspace}
\newcommand{\Operations}{\ensuremath{\mathit{Operations}}\xspace}
\newcommand{\OpExes}{\ensuremath{\mathit{OpExes}}\xspace}
\newcommand{\Events}{\ensuremath{\mathit{Events}}\xspace}
\newcommand{\Histories}{\ensuremath{\mathit{Histories}}\xspace}
\newcommand{\EventOrder}{\ensuremath{\mathit{EventOrder}}\xspace}

\newcommand{\obj}{\ensuremath{\mathit{obj}}\xspace}
\newcommand{\op}{\ensuremath{\mathit{op}}\xspace}
\newcommand{\proc}{\ensuremath{\mathit{proc}}\xspace}

\newcommand{\opex}{\ensuremath{\mathit{opex}}\xspace}
\newcommand{\inv}{\ensuremath{\mathit{inv}}\xspace}
\newcommand{\res}{\ensuremath{\mathit{res}}\xspace}
\newcommand{\ei}{\ensuremath{\mathit{e1}}\xspace}
\newcommand{\eii}{\ensuremath{\mathit{e2}}\xspace}

\newcommand{\idx}{\ensuremath{\mathit{idx}}\xspace}

\newcommand{\type}{\ensuremath{\mathit{type}}\xspace}
\newcommand{\name}{\ensuremath{\mathit{name}}\xspace}
\newcommand{\val}{\ensuremath{\mathit{value}}\xspace}

\newcommand{\correct}{\ensuremath{\mathtt{correct}}\xspace}
\newcommand{\faulty}{\ensuremath{\mathtt{faulty}}\xspace}

\newcommand{\omitting}{\ensuremath{\mathtt{omitting}}\xspace}
\newcommand{\byzantine}{\ensuremath{\mathtt{byzantine}}\xspace}
\newcommand{\byzhist}{\ensuremath{\mathit{ByzHistories}}\xspace}

\newcommand{\normal}{\ensuremath{\mathtt{normal}}\xspace}
\newcommand{\notif}{\ensuremath{\mathtt{notif}}\xspace}

\newcommand{\oo}{\ensuremath{{\rightarrow}}\xspace} 
\newcommand{\eo}{\ensuremath{{<}}\xspace} 

\newcommand{\corr}{\ensuremath{\mathit{P_C}}\xspace}

\newcommand{\cond}{\ensuremath{\mathcal{C}}\xspace}

\newcommand{\OpTermination}{\ensuremath{\mathit{OpTermination}}\xspace}

\newcommand{\PartialOrder}{\ensuremath{\mathit{PartialOrder}}\xspace}
\newcommand{\TotalOrder}{\ensuremath{\mathit{TotalOrder}}\xspace}

\newcommand{\OpHistOrder}{\ensuremath{\mathit{HistoryOrder}}\xspace}
\newcommand{\OpProcOrder}{\ensuremath{\mathit{ProcessOrder}}\xspace}
\newcommand{\OpSetOrder}{\ensuremath{\mathit{SetOrder}}\xspace}
\newcommand{\OpIntOrder}{\ensuremath{\mathit{IntOrder}}\xspace}
\newcommand{\OpFIFOOrder}{\ensuremath{\mathit{FIFOOrder}}\xspace}

\newcommand{\Legality}{\ensuremath{\mathit{Legality}}\xspace}
\newcommand{\validp}{\ensuremath{\mathcal{V}}\xspace}
\newcommand{\safep}{\ensuremath{\mathcal{S}}\xspace}
\newcommand{\livep}{\ensuremath{\mathcal{L}}\xspace}
\newcommand{\ctx}{\ensuremath{\mathit{ctx}}\xspace}
\newcommand{\Safety}{\ensuremath{\mathit{Safety}}\xspace}
\newcommand{\Liveness}{\ensuremath{\mathit{Liveness}}\xspace}
\newcommand{\Validity}{\ensuremath{\mathit{Validity}}\xspace}

\newcommand{\ProcessConsistency}{\ensuremath{\mathit{ProcessConsistency}}\xspace}
\newcommand{\FIFOConsistency}{\ensuremath{\mathit{FIFOConsistency}}\xspace}
\newcommand{\CausalConsistency}{\ensuremath{\mathit{CausalConsistency}}\xspace}
\newcommand{\SequentialConsistency}{\ensuremath{\mathit{SeqConsistency}}\xspace}
\newcommand{\Serializability}{\ensuremath{\mathit{Serializability}}\xspace}

\newcommand{\Linearizability}{\ensuremath{\mathit{Linearizability}}\xspace}
\newcommand{\SetLinearizability}{\ensuremath{\mathit{SetLinearizability}}\xspace}
\newcommand{\IntLinearizability}{\ensuremath{\mathit{IntLinearizability}}\xspace}
\newcommand{\RLinearizability}{\ensuremath{\mathit{RLinearizability}}\xspace}
\newcommand{\RSetLinearizability}{\ensuremath{\mathit{RSetLinearizability}}\xspace}
\newcommand{\RIntLinearizability}{\ensuremath{\mathit{RIntLinearizability}}\xspace}

\newcommand{\Consensus}{\ensuremath{\mathit{C}}\xspace}

\newcommand{\complete}{\ensuremath{\mathit{Complete}}\xspace}

\newcommand{\valence}{\ensuremath{\mathit{Val}}\xspace}

\newcommand{\Continuity}{\ensuremath{\mathit{Continuity}}\xspace}

\newcommand{\Asynchrony}{\ensuremath{\mathit{Asynchrony}}\xspace}
\newcommand{\Resilience}{\ensuremath{\mathit{Resilience}}\xspace}

\newcommand{\NonTriviality}{\ensuremath{\mathit{NonTriviality}}\xspace}
\newcommand{\Termination}{\ensuremath{\mathit{Termination}}\xspace}
\newcommand{\Branching}{\ensuremath{\mathit{Branching}}\xspace}
\newcommand{\PositiveValence}{\ensuremath{\mathit{NonEmptyValence}}\xspace}

\newcommand{\Correctness}{\ensuremath{\mathit{Correctness}}\xspace}

\newcommand{\opp}{\ensuremath{\mathsf{op}}\xspace}

\newcommand{\rread}{\ensuremath{\mathsf{read}}\xspace}
\newcommand{\wwrite}{\ensuremath{\mathsf{write}}\xspace}

\newcommand{\broadcast}{\ensuremath{\mathsf{broadcast}}\xspace}

\newcommand{\rbroadcast}{\ensuremath{\mathsf{r\_broadcast}}\xspace}
\newcommand{\rdeliver}{\ensuremath{\mathsf{r\_deliver}}\xspace}

\newcommand{\id}{\ensuremath{\mathit{id}}\xspace}

\newcommand{\get}{\ensuremath{\mathsf{get}}\xspace}
\newcommand{\set}{\ensuremath{\mathsf{set}}\xspace}

\newcommand{\propose}{\ensuremath{\mathsf{propose}}\xspace}
\newcommand{\decide}{\ensuremath{\mathsf{decide}}\xspace}

\newcommand{\send}{\ensuremath{\mathsf{send}}\xspace}
\newcommand{\receive}{\ensuremath{\mathsf{receive}}\xspace}
\newcommand{\Endround}{\ensuremath{\mathsf{end\_round}}\xspace}

\newcommand{\ie}{\textit{i.e.},\xspace}
\newcommand{\eg}{\textit{e.g.},\xspace}
\newcommand{\etc}{\textit{etc.}\xspace}

\newcommand{\state}{\ensuremath{\sigma}\xspace}
\newcommand{\States}{\ensuremath{\Sigma}\xspace}

\newif\ifannote

\ifannote
    \newcommand{\anninsert}[2]{{\color{#1}#2}}
    \newcommand{\anncomment}[3]{{\color{#1}[#2: #3]}}
\else
    \newcommand{\anninsert}[2]{#2}
    \newcommand{\anncomment}[3]{}
\fi

\newcommand{\TA}[1]{\anncomment{red}{TA}{#1}}

\newcommand{\JW}[1]{\anncomment{brown}{JW}{#1}}

\newcommand{\ta}[1]{\anninsert{red}{#1}}
\newcommand{\af}[1]{\anninsert{violet}{#1}}
\newcommand{\cg}[1]{\anninsert{orange}{#1}}
\newcommand{\nn}[1]{\anninsert{blue}{#1}}
\newcommand{\jw}[1]{\anninsert{cyan}{#1}}

\newcommand{\xOmit}[1]{}
\newcommand{\del}[1]{}

\hideLIPIcs
\nolinenumbers

\title{AMECOS: A Modular Event-based Framework
for Concurrent Object Specification}

\titlerunning{AMECOS: A Modular Event-based Framework for Concurrent Object Specification}

\author{Timoth\'e Albouy}{Univ Rennes, Inria, CNRS, IRISA, Rennes, France}{timothe.albouy@irisa.fr}{https://orcid.org/0000-0001-9419-6646}{}

\author{Antonio Fern\'andez Anta}{IMDEA Networks Institute, Madrid, Spain}{antonio.fernandez@imdea.org}{https://orcid.org/0000-0001-6501-2377}{}

\author{Chryssis Georgiou}{University of Cyprus, Nicosia, Cyprus}{chryssis@ucy.ac.cy}{https://orcid.org/0000-0003-4360-0260}{}

\author{Mathieu Gestin}{Univ Rennes, Inria, CNRS, IRISA, Rennes, France}{mathieu.gestin@inria.fr}{https://orcid.org/0009-0004-4045-6560}{}

\author{Nicolas Nicolaou}{Algolysis Ltd, Nicosia, Cyprus}{nicolas@algolysis.com}{https://orcid.org/0000-0001-7540-784X}{}

\author{Junlang Wang}{IMDEA Networks Institute and Universidad Carlos III de Madrid, Madrid, Spain}{junlang.wang@imdea.org}{https://orcid.org/0009-0003-6004-8823}{}

\authorrunning{Albouy, Fern\'andez Anta, Georgiou, Gestin, Nicolaou, and Wang}

\Copyright{Timothé Albouy, Antonio Fernández Anta, Chryssis Georgiou, Mathieu Gestin, Nicolas Nicolaou, and Junlang Wang}

\EventEditors{Silvia Bonomi, Letterio Galletta, Etienne Rivi\`{e}re, and Valerio Schiavoni}
\EventNoEds{4}
\EventLongTitle{28th International Conference on Principles of Distributed Systems (OPODIS 2024)}
\EventShortTitle{OPODIS 2024}
\EventAcronym{OPODIS}
\EventYear{2024}
\EventDate{December 11--13, 2024}
\EventLocation{Lucca, Italy}
\EventLogo{}
\SeriesVolume{324}
\ArticleNo{11}

\date{}


\funding{
This work has been partially supported by Région Bretagne, the French ANR project ByBloS (ANR-20-CE25-0002-01), the H2020 project SOTERIA, the Spanish Ministry of Science and Innovation under grants SocialProbing (TED2021-131264B-I00) and DRONAC (PID2022-140560OB-I00), the ERDF ``A way of making Europe'', NextGenerationEU, and the Spanish Government's ``Plan de Recuperaci\'on, Transformaci\'on y Resiliencia''.
}


\begin{document}

\maketitle

\begin{abstract}
In this work, we introduce a modular framework for specifying distributed systems that we call AMECOS.
Specifically, our framework departs from the traditional use of sequential specification, which presents limitations both on the specification expressiveness and implementation efficiency of inherently concurrent objects, as documented by Casta\~neda, Rajsbaum and Raynal in CACM 2023.
Our framework focuses on the interactions between the various system components, specified as concurrent objects.
Interactions are described with sequences of object events.
This provides a modular way of specifying distributed systems and separates legality (object semantics) 
from other issues, such as consistency.
We demonstrate the usability of our framework by (i) specifying various well-known concurrent objects, such as registers, shared memory, message-passing, reliable broadcast, and consensus, (ii) providing hierarchies of ordering semantics (namely, consistency hierarchy, memory hierarchy, and reliable broadcast hierarchy), and (iii) presenting a novel axiomatic proof of the impossibility of the well-known Consensus problem.
\end{abstract}



\section{Introduction}
\label{sec:Intro}

{\bf Motivation.} Specifying distributed systems is challenging as they are inherently complex: they are composed of multiple components that \textit{concurrently} interact with each other in unpredictable ways, especially in the face of asynchrony and failures.
Stemming from this complexity, it is very challenging to compose concise specifications of distributed systems and, even further, devise correctness properties for the objects those systems may yield.
%
To ensure the correctness of a distributed system, realized by both \emph{safety} (``nothing bad happens'') and \textit{liveness} (``something good eventually happens'') properties, the specification must capture all of the possible ways in which the system's components can interact with each other and with the system's external environment.
This can be difficult, especially when dealing with complex and loosely coupled distributed systems in which components may proceed independently of the actions of others.
Another challenge caused by concurrency is to specify the {\em consistency} of the system or the objects it implements.
The order in which processes access an object greatly impacts the evolution of the state of said object.
Several types of consistency guarantees exist, from weak ones such as PRAM consistency~\cite{LS88} to stronger ones such as Linearizability~\cite{HW90}.
Therefore, one needs to precisely specify the ordering/consistency guarantees expected by each object the system implements.

To address the inherent complexity of distributed systems, researchers often map executions of concurrent objects to their sequential counterparts, using \emph{sequential specifications}~\cite{AW04,L96}.
Although easier and more intuitive for specifying how abstract data structures behave in a sequential way~\cite{S11}, as noted in~\cite{CRR23}, sequential specifications appear unnatural for systems where events may not be totally ordered, hindering the potential of concurrent objects.
More precisely, there are concurrent objects that do not have a sequential specification (\eg set-linearizable objects or Java's exchanger object~\cite{CRR23}), or objects that, even if they can be sequentially specified, have an inefficient sequential implementation.
For example, it is impossible to build a concurrent queue implementation that eliminates the use of expensive synchronization primitives~\cite{CRR23}.

\smallskip
\noindent {\bf Our approach.} In this work, we propose a modular framework which we call AMECOS (from \textit{A Modular Event-based framework for Concurrent Object Specification}) that {\em does not} use sequential specification of objects, but instead offers a relaxed concurrent object specification mechanism
that encapsulates concurrency intuitively, alleviating the specifier from complex specifications.
Here are some noteworthy {\em features} of our framework.

\noindent
\emph{Component Identification and Interfacing:} Our specification focuses on the {\em interface} between the various system components specified as concurrent objects.
In particular, it considers objects as opaque boxes, specifying them by directly describing the intended behavior and only examining the interactions between the object and its clients. 
In this way, we do not conflate the specification of an object with its implementation, as is typically the case with formal specification languages such as TLA+~\cite{L94} and Input/Output Automata~\cite{LT87}. 
Specifically, these languages specify a distributed system and its components with a state machine (as \jw{a transparent} box). In contrast, our formalism specifies an object at the interface level
(as an opaque box).
Furthermore, we avoid using higher-order logic, which sometimes can be cumbersome, and instead, we use simple logic, rendering our specification ``language'' simple to learn and use.
In some sense, we provide the ``ingredients'' needed for an object to satisfy specific properties and consistency guarantees.

\noindent
\emph{Modularity:} Focusing on the object's interface also provides a {\em modular way} of specifying distributed systems and separates the object's semantics from other aspects, such as consistency.
With our formalism, we can, for example, specify the semantics of a shared register~\cite{L78}, then specify different consistency semantics, such as PRAM~\cite{LS88} and Linearizability~\cite{HW90}, {\em independently}, and finally combine them to obtain PRAM and atomic (\ie linearizable) shared registers, respectively.
This modularity also helps, when convenient (\eg for impossibility proofs), to abstract away the underlying communication medium used for exchanging information.
In fact, we also specify communication media as objects.

\noindent
\emph{Structured Formalism:} The formalism follows a {\em precondition/postcondition} style for specifying an object's semantics, via 3 families of predicates: \textit{Validity}, \textit{Safety}, and \textit{Liveness}.
The first one specifies the requirements for the use of the object (preconditions), while the other two specify the guarantees (hard and eventual) provided by the object (postconditions).
This makes our formalism easy to use, providing a structured way of specifying object semantics.

\noindent
\emph{Notification Operations:} Another feature of the formalism is what we call {\em notification} operations, that is, operations that are not invoked by any particular process, but that spontaneously notify processes of some information.
For example, a broadcast service provides a \broadcast operation for disseminating a message in a system, but all processes of this system must be eventually notified that they received a message without invoking any operation.
So, a notification is a ``callback'' made by an object to a process, and not by a process to an object.
This feature increases our framework's expressiveness compared to formalisms that are restricted to using only invocation and response events for operations~\cite{R18}.

\smallskip
\noindent {\bf Contributions.} 
The following list summarizes the contributions of the paper.

\begin{itemize}
    \item We first present our framework's architecture, basic components, key concepts, and notation (see \Cref{sec:framework}).
    We especially show that our framework can elegantly take into account a wide range of process failures, such as crashes or Byzantine faults (see \jw{\Cref{sec:faults,sec:histories}}).
    Then, we demonstrate how concurrent objects can be specified using histories, preconditions, and postconditions (see \Cref{sec:objspec}).
    
    \item Using our formalism, we show how we can specify several classic consistency conditions, from weak ones such as PRAM
    consistency to strong ones such as Linearizability. Then, we show that we can define consistency conditions (set-linearizability and interval linearizability) for objects that do not have sequential specifications (see \Cref{sec:consistency}).

\item {Using the definitions of object specification and consistency, we show how they can be combined to yield correctness definitions for histories, even in the presence of Byzantine failures (see \Cref{sec:histories})}.  
\item To exemplify the usability of our formalism, we specify shared memory, message passing, and reliable broadcast as concurrent objects (see Section~\ref{sec:obj-spec-examples}).
    The modularity of our formalism is demonstrated by combining the consistency conditions with these object specifications, obtaining shared memory and broadcast hierarchies.
\item 
We demonstrate our framework's effectiveness in constructing axiomatic proofs by presenting a novel, axiomatic proof of the impossibility of resilient consensus objects in an asynchronous system 
    (see \Cref{sec:consensus}).
    To our knowledge, this is one of the simplest and most general proofs of this result. Its simplicity and generality stem from the fact that our formalism abstracts away the implementation details of the object or system being specified, allowing us to focus on proving intrinsic fundamental properties.
    For instance, our impossibility proof abstracts away the communication medium.
    (For completeness, we show in \Cref{sec:asych-axiom-register-channel} that SWSR atomic registers and point-to-point message-passing channels satisfy the relevant assumptions of the proof.)


\end{itemize}
%
We also present a comparison with related work (\Cref{sec:related-work}) and a discussion of our findings (\Cref{sec:conclusion}).

\section{Related Work} \label{sec:related-work}
The present work addresses two different (but not unrelated) problems: object {semantics} specification and consistency conditions.
It also deals with the Consensus impossibility.\vspace{.4em}

\noindent \textbf{Object {semantics} specification.}
As we already discussed, 
traditionally, formal definitions of concurrent objects (\eg shared stacks or FIFO channels) are given by {\em sequential specifications}, which define the behavior of some object when its operations are called sequentially. 
{Distributed algorithms are commonly defined using formal specification languages, \eg input/output (IO) automata~\cite{LT87}, temporal logics (\eg LTL~\cite{P77}, CTL*~\cite{EH83}, TLA~\cite{L94}) and CSP~\cite{AJS05},
for implementing concurrent objects. Formal proofs are used to show that those implementations satisfy the sequential specifications of the object in any possible execution.}


We argue that defining concurrent objects using such formal methods conflates the specification and implementation of said objects.
On the contrary, as we already discussed in \Cref{sec:Intro}, our formalism considers objects as opaque boxes, and the specification stays at the object's interface.
Furthermore, formal methods are typically complex and difficult to learn, requiring specialized tools and expertise.
Our formalism, instead, relies only on simple logic (no higher-order logic is required), making it easy to learn and use.
{To this end, we concur that our formalism {\em complements} existing formal methods by providing an intuitive 
way to express the necessary properties a concurrent object must satisfy. Moreover, the formalism may reveal the necessary 
components (``ingredients'') needed for an object to satisfy specific guarantees. Armed with our object specifications, formal methods 
may be used to specify and compose simpler components, drifting away from the inherent complexities of more synthetic 
distributed structures.}

\noindent \textbf{Consistency conditions.}
Consistency conditions can be seen as additional constraints on the ordering of operation calls that can be applied on top of an object {semantics} specification, which, in this work, we call {\em legality}.
Over time, several very influential consistency conditions have been presented (\eg \cite{ANBKH95,BHG87,HW90,L79,LS88}).
However, all of these consistency conditions have been introduced in their own notations and context (databases, RAM/cache coherence or distributed systems), which raised the need to have a unified formalism for expressing all types of consistency.
Several formalisms have been proposed~\cite{B14,BGYZ14,P16,SN04,VV16}.
\af{We propose a formalism that uses light notation and is very expressive.}
As we demonstrate in \Cref{sec:consistency}, we view legality as the lowest degree of consistency, thus making a clear separation between legality and consistency.
Furthermore, our formalism helps us specify consistency guarantees incrementally, moving from weaker to stronger ones, yielding a consistency hierarchy.
\ta{Several works already presented consistency hierarchies and a way to combine consistencies with object specifications~\cite{P16,VF03}.}
\cg{In contrast to our framework, these approaches rely on sequential object specifications, which, as we have already discussed, can constrain expressiveness.}

Possibly the work in~\cite{VV16} (derived from \cite{B14,BGYZ14}) is the closest to ours with respect to ``specification style.''
However, the object specifications in~\cite{VV16} (and \cite{B14,BGYZ14}) use the artificial notion of \emph{arbitration} to \emph{always} impose a total order on object operation executions, whereas our formalism does not require a total order (unless the consistency imposes it); in general, we only consider partial orders of operation executions. 
Another notable difference with \cite{VV16} is that their consistency specification is for storage objects, whereas we specify consistency conditions in general, and then we combine them with a specification of a shared memory object to yield a consistency hierarchy for registers (similar to storage objects considered in~\cite{VV16}).
Additionally, compared to other endeavors such as \cite{B14,BGYZ14,P16}, our framework also considers strong failure types such as Byzantine faults.

\noindent \textbf{Impossibility of Asynchronous Consensus.}
Consensus is a fundamental abstraction of distributed computing with a simple premise: all participants of a distributed system must propose a value, and all participants must eventually agree on one of the values that have been proposed~\cite{L96}.
But just as fundamental is the impossibility theorem associated with consensus in the presence of asynchrony and faults.
This result of impossibility, colloquially known as the \emph{FLP} theorem (for the initials of its authors), was first shown in 1983~\cite{FLP85}.
Later on, different approaches for proving similar theorems were proposed (\eg \cite{GL23,HRT98,HS97}).
Notably, the impossibility of asynchronous resilient consensus can be proved using algebraic topology and, more specifically, the asynchronous computability theorem~\cite{HS99}.
\ta{Constructive proofs follow another interesting approach \cite{GL23,V04}: they explicitly describe how a non-terminating execution of consensus can be constructed.\footnote{
    In their paper~\cite{GL23}, Gafni and Losa show the equivalence of 4 different models in terms of computability power: asynchronous 1-resilient atomic shared memory, asynchronous 1-resilient message passing, synchronous fail-to-send message passing, and synchronous fail-to-receive message passing.
    They then present a constructive impossibility proof in the synchronous fail-to-send message-passing model, and the impossibility in the other models directly follows.
    However, unlike ours, their proof still makes assumption of the communication medium.
}}
{Like \cite{T91}, our proof follows an axiomatic approach: the notion of asynchronous resilient consensus is defined as a system of axioms, and then it is proved that this system is inconsistent, \ie that there is a contradiction.}

Compared to previous impossibility proofs of asynchronous resilient consensus, we believe our proof to be one of the most general, partly due to the more natural notations offered by our specification formalism.
In particular, unlike~\cite{FLP85}, which assumes that processes communicate through a message-passing network, our proof is agnostic on the communication medium (as long as such communication medium is asynchronous), hence it holds both for systems using send/receive \af{or RW shared memory.}
In addition, our proof is more general than many previous proofs, in the sense that it shows an impossibility for a very weak version of the problem.

For instance, our proof differs from~\cite{FLP85,HS99}, which assumes that, at least in some specific cases, the value decided by the consensus instance is a value proposed by some process.
In contrast, our proof does not make this assumption, as it does not need to relate the inputs (proposals) to the outputs (decisions) of the consensus execution.
\JW{REMOVED:This way, our proof also shows the impossibility of other abstractions, such as shared random coins, which allow the system processes to agree on some value chosen at random.}

\begin{figure}[t!]
    \centering
    \includegraphics[scale=0.45]{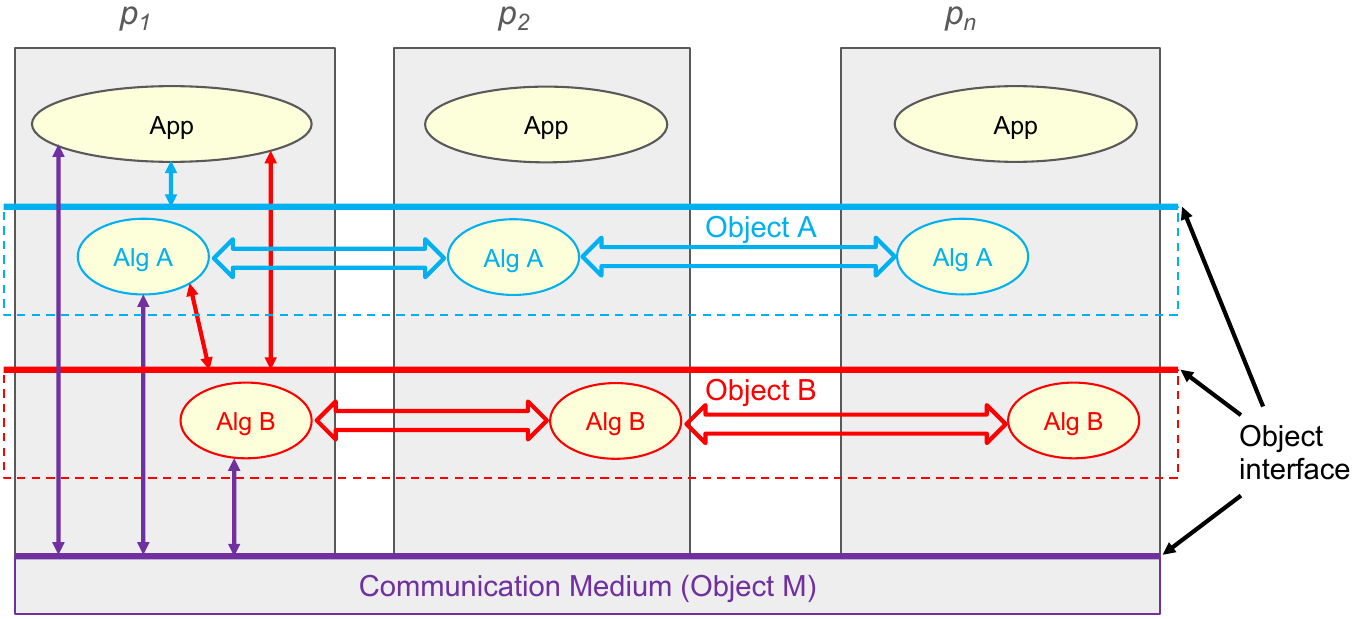}
    \caption{Example of a distributed system architecture with 3 objects. The distributed system is formed by $n$ processes that communicate using a Communication Medium, which is Object $M$. The system has 2 more objects, Objects $A$ and $B$, built by local \jw{modules} (Alg.~$A$ and $B$) in each process, that cooperate using a protocol (horizontal arrows) to implement the object. Each process has an application algorithm that can execute the operations (vertical arrows) of the objects offered by their interface. Also, objects' local \jw{modules} can execute operations of other objects. }
    \label{fig:sys-archi}
\end{figure}

\section{Framework: Architecture, Components, Notation, and Concepts} \label{sec:framework}

\subsection{System Architecture}

The proposed framework assumes a distributed system as depicted in \Cref{fig:sys-archi}, which has a set of processes that interact with some Communication Medium (\eg shared memory, message passing), modeled as an object. The processes have applications and local \jw{modules} that implement other objects. \cg{(Processes are essentially computing entities, and they could be modeled as I/O Automata~\cite{LT87}.)} Each object offers an interface formed by operations that can be invoked by applications and \jw{modules} of other objects. An execution of this system is an execution of the applications of all the processes and the execution of the object operations these applications invoke (directly or by execution of operations in other objects).\vspace{-.5em}

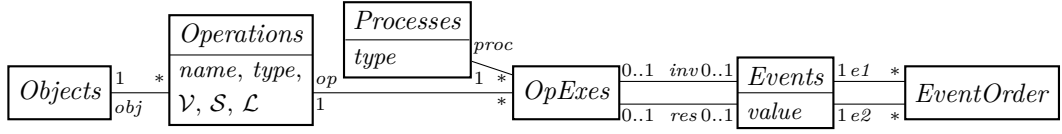
\begin{figure}[t!]
    \centering
    \begin{tikzpicture}
    
\tikzmath{
\xObj = -6.7;
\xOpe = -4.3;
\xProc = -2.1;
\xOpex = 0;
\xEvt = 2.9;
\xOrd = 5.5;
\yProc = .7;
}

\newcommand{\titleSize}{\normalsize}

\draw (\xOpex, .15) -- (\xEvt,.15);
\node at (\xEvt-.9, .3) {\scriptsize 0..1};
\node at (\xEvt-1.35, .3) {\scriptsize \inv};
\node at (\xOpex+.95, .3) {\scriptsize 0..1};

\draw (\xOpex, -.15) -- (\xEvt,-.15);
\node at (\xEvt-.9, -.3) {\scriptsize 0..1};
\node at (\xEvt-1.35, -.32) {\scriptsize \res};
\node at (\xOpex+.95, -.3) {\scriptsize 0..1};

\draw (\xOrd, .15) -- (\xEvt,.15);
\node at (\xEvt+.75, .3) {\scriptsize 1};
\node at (\xEvt+1, .3) {\scriptsize \ei};
\node at (\xOrd-1.15, .25) {\scriptsize *};

\draw (\xOrd, -.15) -- (\xEvt,-.15);
\node at (\xEvt+.75, -.3) {\scriptsize 1};
\node at (\xEvt+1, -.3) {\scriptsize \eii};
\node at (\xOrd-1.15, -.35) {\scriptsize *};

\draw (\xOpe, 0) -- (\xObj, 0);
\node at (\xObj+.8, .2) {\scriptsize 1};
\node at (\xObj+.9, -.2) {\scriptsize \obj};
\node at (\xOpe-1.08, .15) {\scriptsize *};

\draw (\xOpex, 0) -- (\xOpe, 0);
\node at (\xOpe+1.05, -.15) {\scriptsize 1};
\node at (\xOpe+1.12, .15) {\scriptsize \op};
\node at (\xOpex-.85, -.15) {\scriptsize *};

\node[draw,thick,fill=white,minimum height=.7cm] (object) at (\xObj, 0) {\titleSize\Objects};

\node[draw,thick,fill=white,align=left,minimum height=.7cm] (operation) at (\xOpe, .3) {\titleSize\Operations\\\small\name, \type,\\\small\validp, \safep, \livep};
\draw (\xOpe-.94, .5) -- (\xOpe+.94,.5);

\node[draw,thick,align=left,minimum height=.7cm] (process) at (\xProc,\yProc) {\titleSize\Processes\\\small\type};
\draw (\xProc-.82, \yProc) -- (\xProc+.82,\yProc);

\node[draw,thick,fill=white,minimum height=.7cm] (opex) at (\xOpex, 0) {\titleSize\OpExes};

\node[draw,thick,fill=white,align=left,minimum height=.7cm] (event) at (\xEvt, 0) {\titleSize\Events\\\small\val};
\draw (\xEvt-.63, 0) -- (\xEvt+.63,0);

\node[draw,thick,fill=white,minimum height=.7cm] (order) at (\xOrd, 0) {\titleSize\EventOrder};

\draw (opex) -- (process);
\node at (\xProc+.95, \yProc-.5) {\scriptsize 1};
\node at (\xProc+1.13, \yProc-.1) {\scriptsize \proc};
\node at (\xOpex-.85, .15) {\scriptsize *};
    
\end{tikzpicture}
    \vspace{-2em}
    \caption{Class diagram for the relations and attributes of the components (sets) of the framework.}
    \label{fig:data-model}
\end{figure}

\subsection{Components} \label{sec:model-sets}
\af{In the framework, an execution of a distributed system is described with 6 (potentially infinite) sets.} As shown in \Cref{fig:data-model}, these sets are in relation to each other.

    \noindent
    \textbf{\Processes:} contains all \textit{processes} $p_i$ (of identity $i$) of the system \ta{execution}, where the $p_i.\type \in \{\correct, \faulty\}$ attribute is the \textit{type} of $p_i$ (either correct or faulty, see \Cref{sec:faults}). 
    
    \noindent
    \textbf{\Objects:} contains all the \textit{objects} of the system (\eg a channel, a register, a stack, \etc).
    An object is associated with a set of {\em operations} (the interface) that processes can \af{execute on} it.

    \noindent
    \textbf{\Operations:} contains all the \textit{operations} $\op$ that can be \af{executed} on some object $\op.\obj \in \Objects$, where the $\op.\name$ attribute is the name of \op (\eg \wwrite or \rread), the $\op.\type \in \{\normal, \notif\}$ attribute is the type of $\op$ (either it is a normal operation or a notification operation, see \Cref{sec:notations}), and the $\op.\validp$, $\op.\safep$, and $\op.\livep$ attributes are predicates respectively defining the invocation validity, safety, and liveness of the operation (see \Cref{sec:objspec}).

    \noindent
    \textbf{\OpExes:} contains all \textit{operation executions} (\textit{op-ex} for short) $o$ of an operation $o.\op \in \Operations$ by a process $o.\proc \in \Processes$. When an operation is \af{executed}, it produces an invocation event and a response event. Hence, an op-ex $o$ is the pair $(o.\inv, o.\res)$ of the invocation and the response
    events of $o$. \af{If $o.\op.\type = \notif$ then op-ex $o$ has no invocation event, \ie $o.\inv=\bot$ (see \Cref{sec:notations}).}

    \noindent
    \textbf{\Events:} contains all the \textit{events} $e$ in \OpExes, where the $e.\val$ attribute is the value of the event (\ie the input or output value of an op-ex, see \Cref{sec:notations}).
    The set \Events does not contain the $\bot$ value, which denotes the absence of an event.
    
    \noindent
    \textbf{\EventOrder:} corresponds to the \af{total order of the events in} \Events, represented as a set of event pairs $\mathit{ep}=(\mathit{ep}.\ei,\mathit{ep}.\eii)$.
    \af{For the sake of simplicity, we denote this total order with $\eo$ over all elements of \Events as follows:}
    $\eo \triangleq \{(\mathit{ep}.\ei,\mathit{ep}.\eii) \mid \Forall \mathit{ep} \in \EventOrder\}.$
    The order $\eo$ defines the temporal ordering of events, \ie $e \eo e'$ means event $e$ happens before event $e'$.
%

\subsection{Notation} \label{sec:notations}

As described,
an op-ex $o \in \OpExes$ is a pair $(i,r)$ where $i=o.\inv$ and $r=o.\res$.
An op-ex $o \in \OpExes$ can have the following configurations: 
\begin{itemize}
    \item $o=(i,r)$, where $i,r \in \Events$, \af{$i \neq r$, and $i < r$:} in this case, $o$ is said to be a {\em complete} op-ex (an operation invocation that has a matching response),

    \item $o=(i,\bot)$, where $i \in \Events$: in this case, $o$ is said to be a {\em pending} op-ex (this notation is useful to denote operation calls by faulty processes or operation calls that do not halt),

    \item $o=(\bot,r)$, where $r \in \Events$: in this case, $o$ is said to be a {\em notification} op-ex (\ie its operation is not a callable operation, but an operation spontaneously invoked by the object to transmit information to its client).
\end{itemize}
%

The ``$\equiv$'' relation indicates that some op-ex $o$ follows a given form or the same form as another op-ex: an op-ex $o$ could be of the form ``$\rread$ op-ex on register $R$ by process $p_i$ which returned \af{output} value $v$'', which we denote ``$o \equiv R.\rread_i()/v$''.
If the object, process ID, parameter, or return value are not relevant, we omit these elements in the notation: the form ``$L.\get_i(j)/v$ op-ex on list $L$, process $p_i$, index $j$ \af{as input value} and \af{returned output} value $v$'' could simply be written ``$\get()$''.
If only some (but not all) parameters can have an arbitrary value, we can use the ``$-$'' notation: the form ``$S.\set(k,v)$ op-ex on key-value store $S$ for key $k$ and for any value $v$'' could be written ``$S.\set(k,-)$''.
\ta{As notifications have no input, the $()$ parameter parentheses are omitted for notifications: the form ``$\receive$ notification op-ex of message $m$ by receiver process $p_i$ from sender process $p_j$'' is denoted ``$\receive_i/(m,j)$''.}
Lastly, we denote by ``\opp'' any \af{op-ex operation:} the form ``any op-ex on register $R$ (\wwrite or \rread)'' could thus be written ``$R.\opp()$''.
Pending op-exes and complete op-exes with no return value can be respectively denoted with a $\bot$ and $\varnothing$ symbol \af{as their output value.}
For instance, the forms ``All pending op-exes'' and ``All complete op-exes with no return value'' can be written ``$\opp()/\bot$'' and ``$\opp()/\varnothing$'', respectively.
By abuse of notation, to refer to any op-ex of a set $O$ that follows a given form $f$, we can write $f \in O$, for example $R.\wwrite(v) \in O$.

\subsection{Histories}

\af{The six sets of the framework are used to describe an execution of a distributed system. We are interested in describing all possible executions of the system for a given set of \Objects and \Operations.
Hence, each such system execution is described with \Events, \OpExes, \EventOrder, and \Processes. We call this a \emph{history}. Note that a history captures a system execution via the events and op-exes observed in the objects' interfaces.}

Hence, a history of a distributed system is a tuple $H=(E,\eo,O, P)$, where $E=\Events$ is the set of events, $\eo=\EventOrder$ is a \emph{strict total order} of events in $E$, $O=\OpExes$ is the set of op-exes, and $P=\Processes$ is the set of processes.
There are some natural constraints on a history that cannot be expressed directly on the diagram of \Cref{fig:data-model}.
\begin{itemize}
    \item \textbf{Event validity}: Every event must be part of exactly one op-ex.
    \item \textbf{Operation validity}: If an operation is a notification, then all its op-exes must be notification op-exes, otherwise, they must all be complete or pending op-exes.
\end{itemize}

\af{
In the sequel, we will often consider subhistories of a history $H=(E,\eo,O, P)$. 
For instance, it is useful to consider the subhistory obtained by projecting a history to one particular object. Hence, given history $H=(E,\eo,O,P)$ and object $x$, we denote by $H|x$ the subhistory containing only the events of $E$ applied to $x$ and the op-exes of $O$ applied to $x$, and the corresponding subset of \eo.
The concepts of legality, consistency, and correctness defined in the next sections can hence be applied to histories and subhistories.}


\subsection{Process Faults} \label{sec:faults}

Our framework considers two very general types of process failures: \emph{omission faults} and \emph{Byzantine faults}.
In the framework's model presented in \Cref{fig:data-model}, for any $p \in \Processes$, omission faults concern $p$ only if $p.\type = \faulty.\omitting$, and Byzantine faults concern $p$ only if $p.\type = \faulty.\byzantine$.
Processes $p$ of type $p.\type = \correct$ are subject to none of these faults.

Omissions correspond to missing events, like op-ex invocations that do not have a matching response, for whatever reason, producing pending op-exes.
We assume that such omitting processes, although they may suffer omissions at any time, follow their assigned algorithm.
A {\em crash fault} is a special case of omission fault on a process $p$, where there exists a point $\tau$ in the sequence of events of the history (the crash point) such that, 
$p$ has no event after $\tau$ in the sequence.
\af{Observe that omission faults also account for benign dynamic process failures like crash-recovery models.}

On the other hand, Byzantine processes may arbitrarily deviate from their assigned algorithm (for instance, because of implementation errors or attacks).
Strictly speaking, given that their behavior can be arbitrary, we cannot say that Byzantine processes execute actual op-exes on the same operations and objects as non-Byzantine processes.
For instance, Byzantine processes may simulate op-exes that can appear legitimate to other processes. We further discuss Byzantine faults in \Cref{sec:histories}.

Observe that by adding more constraints to the model, new failure subtypes can be derived from these two basic failure types.


\section{Object Specification and History Legality} \label{sec:objspec}

\subsection{Object Specification}

\af{In our formalism, we specify an object using a set of conditions that are defined for the operations and applied to the op-exes of the} object.
There are two types of such conditions (that we express as predicates): preconditions (invocation validity) and postconditions (safety and liveness).
Every operation of every object has a \validp precondition and two \safep and \livep postconditions (if not given, they are assumed to be satisfied).
We will use the register object as an example to understand better the notations defined below.
A register $R$ is associated with two operations, $R.\rread()/v$ and $R.\wwrite(v)$, where the former returns the value of $R$ and the latter sets the value of $R$ to $v$, respectively.


\smallskip
\noindent{\bf Op-ex context.}
The context of an op-ex $o$ is the set of all op-exes preceding $o$ \emph{in the same object} with respect to a binary relation \oo defined over $O$.

\begin{definition}[Context of an op-ex]
The \emph{context} of an op-ex $o \in O$ with respect to a binary relation \oo over $O$ is defined as
$
    \ctx(o,O,\oo) \triangleq (O_c,\oo_c),
$
where $O_c \triangleq \{o' \in O \mid o' \oo o, o.\op.\obj=o'.\op.\obj\}$ and $\oo_c\, \triangleq \{(o',o'') \in \,\oo\, \mid o',o'' \in O_c \cup \{o\}\}$.
\end{definition}
For instance, the context of a $R.\rread()/v$ op-ex is made of all the previous op-exes of register $R$ with respect to a given \oo relation.
Note that pending op-exes can be part of the context of other op-exes, and thus influence their behavior (and especially their return value in the case of complete op-exes or notifications).\smallskip

\noindent{\bf Preconditions.}
The \emph{preconditions} of an object are the use requirements of this object by its client that are needed to ensure that the object works properly.
Typically, a precondition for the operation on an object can require that the input (parameters) of this operation is valid, or that some op-ex required for this operation to work indeed happened before.
For instance, we cannot have a $\rread$ op-ex in register $R$ if there was no preceding $\wwrite$ op-ex in $R$.
Another example of precondition for the $\mathsf{divide}(a,b)/d$ operation that returns the result $d$ of the division of number $a$ by number $b$, is that $b$ must not be zero.
These preconditions are given for each operation of an object by the invocation validity predicate \validp.
\begin{definition}[Invocation validity predicate \validp]
Given an operation $\opp \in \Operations$, its \emph{invocation validity predicate} $\opp.\validp(o,\ctx(o,O,\oo))$ indicates whether an op-ex $o=(i{\neq}\bot,r) \in \OpExes$ of \opp (\ie $o.\op \equiv \opp$) respects the usage contract of the object 
given its context $\ctx(o,O_c,\oo_c)$.
\end{definition}
%

\noindent{\bf Postconditions.}
The \emph{postconditions} of an object are the guarantees provided by this object to its client.
The postconditions are divided into two categories: \emph{safety} and \emph{liveness}.
Broadly speaking, safety ensures that nothing bad happens, while liveness ensures that something good will eventually happen.
For a given op-ex $o$, safety is interested in the prefix of op-exes of $o$ (\ie its context), while liveness is potentially interested in the whole history of op-exes. 
For example, for a register object $R$, the safety condition is that the value $v$ returned by a $R.\rread()/v$ is (one of) the latest written values, while the liveness condition is that the $R.\rread()$ and $R.\wwrite()$ op-exes always terminate.
These postconditions are given for each operation of an object by the safety predicate \safep and the liveness predicate \livep.

\begin{definition}[Safety predicate \safep]
\sloppy
Given an operation $\opp \in \Operations$, its \emph{safety predicate} $\opp.\safep(o,\ctx(o,O,\oo))$ indicates whether $r.\val$ is a valid return value for op-ex $o=(i,r{\neq}\bot) \in \OpExes$ of \opp (\ie $o.\op \equiv \opp$) 
in relation to its context
$\ctx(o,O_c,\oo_c)$.
\end{definition}

We can see in the above definition that the $\opp.\safep(o,\ctx(o,O,\oo))$ predicate is not defined if $o=(i,\bot)$, that is, if $o$ is a pending op-ex. 

\begin{definition}[Liveness predicate \livep]
Given an operation $\opp \in \Operations$, its \emph{liveness predicate} $\opp.\livep(o,H,\oo)$ indicates whether an op-ex $o=(i,r) \in \OpExes$ of \opp (\ie $o.\op \equiv \opp$) respects the liveness specification of \opp.
\end{definition}

\noindent{\bf Example.} 
As a complete example, the following is the specification of a single-writer single-reader (SWSR) shared register $R$ using the above conditions. Let $p_w$ and $p_r$ be the writer and reader processes, respectively. Recall that $\ctx(o,O,\oo)=(O_c,\oo_c)$ is the context of op-ex $o$. \af{Let us define predicate $\OpTermination(o) \triangleq (o.\proc.\type = \correct) \implies (o \not\equiv \opp()/\bot)$, which forces an op-ex $o$ to have a response if executed by a correct process.}

\paragraph{Operation \rread.} If $o \equiv R.\rread()/v$, then we have the following conditions.
\begin{align*}
    \rread.\validp(o,(O_c,\oo_c)) &\triangleq (o.\proc=p_r) \land (\Exists o' \equiv R.\wwrite(-) \in O_c). \\
    \rread.\safep(o,(O_c,\oo_c)) &\triangleq v \in \{v' \mid \Exists o' {\equiv} R.\wwrite(v') \in O_c, \Nexists o'' {\equiv} R.\wwrite(-) \in O_c': o' \oo_c o''\}. \\
    \rread.\livep(o,H,\oo) &\triangleq \OpTermination(o).
\end{align*}

The \validp predicate states that only the reader process can read and the register must have been previously written.
The \safep predicate states that a read must return one of the latest written values.
\af{The \livep predicate prevents a \rread op-ex of a correct process to be pending.}

\paragraph{Operation \wwrite.} If $o \equiv R.\wwrite(v)$, then we have the following condition.
\begin{align*}
    \wwrite.\validp(o,(O_c,\oo_c)) &\triangleq o.\proc=p_w.
    \hspace{3em}
    \wwrite.\livep(o,H,\oo) \triangleq \OpTermination(o).
\end{align*}
%
Observe that the \safep predicate is not provided and is hence assumed to be always satisfied.

\subsection{History Legality}

We now define the notion of history legality.
\begin{definition}[\af{History validity, safety, and liveness}]
Given a history $H=(E,\eo,O,P)$ and a relation \oo on $O$, the following predicates define the \emph{validity}, \emph{safety}, and \emph{liveness} of $H$.
\begin{align*}
    \Validity(H,\oo) &\triangleq \Forall o=(i{\neq}\bot,r) \in O: o.\op.\validp(o,\ctx(o,O,\oo)). \\
    \Safety(H,\oo) &\triangleq \Forall o=(i,r{\neq}\bot) \in O: o.\op.\safep(o,\ctx(o,O,\oo)). \\
    \Liveness(H,\oo) &\triangleq \Forall o=(i,r) \in O: o.\op.\livep(o,H,\oo).
\end{align*}
\end{definition}

Notice that, for an op-ex $o$, if $o$ is a notification, we do not need to verify its invocation validity, and if $o$ is a pending op-ex, we do not need to verify its safety.

\begin{definition}[Legality condition]
\label{def:legality}
Given a history $H=(E,\eo,O,P)$ and a relation \oo on $O$, the \emph{legality condition} is defined as
    $$\Legality(H,\oo) \triangleq \{\Validity(H,\oo), \Safety(H,\oo), \Liveness(H,\oo)\}.$$
\end{definition}
\Legality is defined as a set of clauses (or constraints) on a history $H$ and an op-ex relation \oo.
Informally, a history $H$ is legal if and only if all clauses of $\Legality(H,\oo)$ evaluate to \ttrue.

\section{Consistency Conditions}
\label{sec:consistency}

\af{In the previous section we have presented how object legality can be specified using operation conditions, abstracting the consistency model with a binary order relation \oo. In this section we describe how to define order relations \oo to extend legality with different consistency conditions.}
We first define reusable predicates describing certain constraints on the op-ex order \oo (\Cref{sec:opex-orders}) and then we define common consistency conditions (\Cref{sec:cons-cond}) to showcase the power of the formalism. In addition, we provide the definitions of set-linearizability~\cite{N94} and interval-linearizability~\cite{CRR18}, which are new interesting consistency conditions; objects with these consistencies do not have sequential specifications~\cite{CRR23} (\Cref{sec:set-int-lin}).

\subsection{Op-ex Order Relations} \label{sec:opex-orders}

We first define partial and total orders within our framework.

\begin{definition}[Generic strict orders]
\label{def:generic-orders}
Given an arbitrary set $S$ and an arbitrary binary relation $\prec$ over the elements of $S$, the following predicates define \emph{strict partial order} and \emph{strict total order}.
\begin{align*}
    \PartialOrder(S,\prec) \triangleq& \,(\Forall e \in S: e \not\prec e ) \land (\Forall e,e',e'' \in S: e \prec e' \prec e'' \implies e \prec e''). \\
    \TotalOrder(S,\prec) \triangleq& \,\PartialOrder(S,\prec) \land (\Forall e,e' \in S, e \neq e': e \prec e' \lor e' \prec e).
\end{align*}
\end{definition}

We call these orders ``strict'' because they are irreflexive.
Observe that the asymmetry property is redundant for strict orders because it directly follows from irreflexivity and transitivity.
As we can also see, total order adds the connectedness property to partial order.

\begin{definition}[Basic orders]
Given a history $H=(E,\eo,O,P)$ and 
a relation \oo on the set of op-exes $O$, 
the following predicates define \emph{history order}, \emph{process order} and \emph{FIFO order}:\vspace{-.6em}
\begin{align*}
    \OpHistOrder(H,\oo) &\triangleq \Forall o=(i,r),o'=(i',r') \in O, r \neq \bot: 
    \\&\hspace{1em} 
    ((i' \neq \bot \land r \eo i') \lor (i'=\bot \land r \eo r')) \implies (o \oo o' \land o' \not\!\!\oo o). \\
    \OpProcOrder(H,\oo) &\triangleq \Forall p_i \in P:
    \OpHistOrder(H|p_i,\oo) \land \TotalOrder(O|p_i,\oo). \\
    \OpFIFOOrder(H,\oo) &\triangleq \Forall o_i,o_i' \in O|p_i, o_j,o_j' \in O|p_j: 
    \\&\hspace{1.3em}
    (o_i \oo o_i' \oo o_j \oo o_j' \land o_i \oo o_j') \implies (o_i \oo o_j \land o_i' \oo o_j').
\end{align*}
\end{definition}\vspace{-.1em}

Note that the above predicates do not define ``classic'' order relations (strict or not) per se, as they do not guarantee all the required properties.
These predicates define how a ``visibility'' relation \oo between op-exes of history $H$ should look in different contexts, in the sense that the behavior of an op-ex is determined by the set of op-exes it ``sees''.
\begin{itemize}
    \item In \OpHistOrder, we check if \oo respects the event order \eo: if two op-exes are not concurrent with respect to the \eo order, then the oldest one must precede the newest one.
    Note that we distinguish whether the second op-ex $o'$ is a notification or not.

    \item In \OpProcOrder, we check if \oo totally orders the op-exes of each process $p_i$ while also respecting the event order of $p_i$.

    \item In \OpFIFOOrder, we check that, if an op-ex of any given process sees some other op-ex of another process, then it also sees all the previous op-exes of the latter process.
    Furthermore, we also check that the set of op-exes seen by the op-exes of a given process is monotonically increasing, \ie that a given op-ex sees all the op-exes that its predecessors (of the same process) saw.
    More details about \OpFIFOOrder can be found in \Cref{sec:fifo-addendum}.
\end{itemize}

\subsection{Classic Consistency Conditions} \label{sec:cons-cond}
This section defines common consistency conditions from the literature.
\ta{They are represented as sets of clauses, extending those in \Legality and constraining the \oo op-ex order.}

\sloppy
\begin{definition}[Classic consistency conditions] \label{def:classic-consistency}
Given a history $H=(E,\eo,O,P)$ and a relation \oo on $O$, the following sets of clauses respectively define \emph{process consistency}, \emph{FIFO consistency~\cite{LS88}}, \emph{causal consistency~\cite{ANBKH95,PMJ16}}, \emph{serializability~\cite{BHG87}}, \emph{sequential consistency~\cite{L79}} and \emph{linearizability~\cite{HW90}}.

\begin{align*}
    \ProcessConsistency(H,\oo) &\triangleq \Legality(H,\oo) \cup \{\OpProcOrder(H,\oo)\}. \\
    \FIFOConsistency(H,\oo) &\triangleq \ProcessConsistency(H,\oo) \cup \{\OpFIFOOrder(H,\oo)\}. \\
    \CausalConsistency(H,\oo) &\triangleq \FIFOConsistency(H,\oo) \cup \{\PartialOrder(O,\oo)\}. \\
    \Serializability(H,\oo) &\triangleq \Legality(H,\oo) \cup \{\TotalOrder(O,\oo)\}. \\
    \SequentialConsistency(H,\oo) &\triangleq \Serializability(H,\oo) \cup \CausalConsistency(H,\oo). \\
    \Linearizability(H,\oo) &\triangleq \SequentialConsistency(H,\oo) \cup \{\OpHistOrder(H,\oo)\}.
\end{align*}
\end{definition}

Note that the above \Serializability condition strays away from the traditional definition of serializability, as it considers that a transaction (originally defined as an atomic sequence of op-exes~\cite{BHG87}) is the same as a single op-ex.
Likewise, as discussed in detail in \Cref{sec:fifo-addendum}, our definition of \FIFOConsistency differs from the traditional definition of PRAM consistency.

\begin{figure}[t]
    \centering
    \begin{tikzpicture}
    
\tikzmath{
\xOffset=3.2;
\yOffset=.75;
}

\newcommand{\fontSize}{\footnotesize}


\node[draw] (leg) at (0,\yOffset) {\fontSize\Legality};


\node[draw] (ser) at (\xOffset,\yOffset) {\fontSize\Serializability};

\node[draw] (seq) at (2*\xOffset,\yOffset) {\fontSize\SequentialConsistency};

\node[draw] (lin) at (3*\xOffset,\yOffset) {\fontSize\Linearizability};


\node[draw] (proc) at (0,0) {\fontSize\ProcessConsistency};

\node[draw] (fifo) at (\xOffset,0) {\fontSize\FIFOConsistency};

\node[draw] (caus) at (2*\xOffset,0) {\fontSize\CausalConsistency};

\node[draw] (intlin) at (\xOffset,2*\yOffset) {\fontSize\IntLinearizability};

\node[draw] (setlin) at (2*\xOffset,2*\yOffset) {\fontSize\SetLinearizability};

\draw[rounded corners,ultra thick,red,dashed] (\xOffset-1.2,\yOffset+.35) rectangle (3*\xOffset+1.2,\yOffset-.35) {};
\node[red] at (3*\xOffset,.2) {\fontSize\bf Total order};

\draw[->] (leg) -- (proc);
\draw[->] (leg) -- (ser);
\draw[->] (proc) -- (fifo);
\draw[->] (fifo) -- (caus);
\draw[->] (caus) -- (seq);
\draw[->] (ser) -- (seq);
\draw[->] (seq) -- (lin);

\draw[->] (leg) -- (0,2*\yOffset) -- (intlin);
\draw[->] (intlin) -- (setlin);
\draw[->] (setlin) -- (3*\xOffset,2*\yOffset) -- (lin);
    
\end{tikzpicture}
    \vspace{-.5em}
    \caption{Hierarchy of the consistencies defined in this paper and their relative strengths \cite{P16}.}\vspace{-.7em}
    \label{fig:cons-graph}
\end{figure}
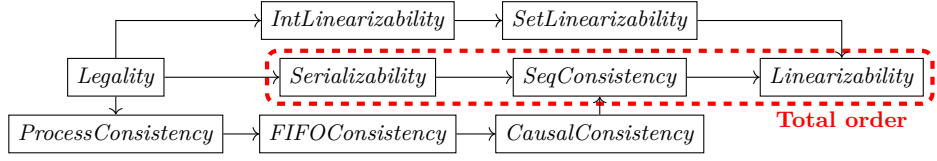

We illustrate in \Cref{fig:cons-graph} the relations between all the consistency conditions defined in this section.
In this figure, if we have $\cond_1 \rightarrow \cond_2$ for two consistency conditions $\cond_1$ and $\cond_2$, then it means that $\cond_2$ is stronger than $\cond_1$, and thus, that $\cond_2$ imposes more constraints on the order of op-exes.
The conditions inside the red rectangle are conditions that impose a total order of op-exes.
Combining these consistency conditions with other object specifications allows us to obtain multiple consistent object specifications (see \Cref{sec:obj-cons-examples}).

\subsection{Set- and Interval-Linearizability Definitions} \label{sec:set-int-lin}

This section adds set-linearizability~\cite{N94} and interval-linearizability~\cite{CRR18} to the repertoire of consistency conditions supported by the framework. These are interesting because they define objects with no sequential specification~\cite{CRR23}.
We first define interval and set orders.

\begin{definition}[Set and interval orders]
Given a history $H=(E,\eo,O,P)$ and a relation \oo on $O$, the following predicates define \emph{interval} and \emph{set orders}:
\begin{align*}
    \OpIntOrder(H,\oo) &\triangleq (\Forall o \in O: o \not\!\!\oo o) \land (\Forall o,o' \in O, o \neq o': o \oo o' \lor o' \oo o) \\
    &\hspace{.3em}\land (\Forall o,o',o'' \in O: o \oo o'' \implies (o \oo o' \lor o' \oo o'')). \\
    \OpSetOrder(H,\oo) &\triangleq \OpIntOrder(H,\oo)
    \land (\Forall o,o',o'' \in O, o \neq o'': o \oo o' \oo o'' \implies o \oo o'').
\end{align*}
\end{definition}


    In \OpIntOrder, an op-ex is represented as a time interval, and we check that it can see only all op-exes with which it overlaps, and all previous op-exes.
    The first clause guarantees irreflexivity (an op-ex cannot see itself), the second connectedness (all op-exes are in relation with each other), and the last one ensures that no forbidden pattern is present.

    In \OpSetOrder, we check that an op-ex can see only all other op-exes of its equivalence class (except itself), and all previous op-exes.
    In addition to \OpIntOrder, \OpSetOrder guarantees a weakened version of transitivity, allowing two-way cycles between two or more op-exes, thus creating equivalence classes.
    Let us remark that the weakened transitivity property of \OpSetOrder implies the last clause of \OpIntOrder.

Leveraging the above order relations, we can define set- and interval-linearizability.

\begin{definition}[Set- and interval-linearizability] \label{def:additional-consistency}
Given a history $H=(E,\eo,O)$ and a relation \oo on $O$, the following predicates define \emph{set-linearizability~\cite{N94}} and \emph{interval-linearizability~\cite{CRR18}}.
\begin{align*}
    &\IntLinearizability(H,\oo) \triangleq \Legality(H,\oo) \;\cup
    \{\OpHistOrder(H,\oo), \OpIntOrder(H,\oo)\}. \\
    &\SetLinearizability(H,\oo) \triangleq  \IntLinearizability(H,\oo) \cup \{\OpSetOrder(H,\oo)\}.
\end{align*}
\end{definition}


\begin{figure}[t]
\centering
\begin{minipage}{.48\textwidth}
    \centering
    \begin{tikzpicture}
    
\tikzmath{
\xLength=6;
\yOffset=.6;
\circleR=.2;
\xLin1 = 1.8;
\xLin2 = 4.3;
}

\newcommand{\fontSize}{\footnotesize}

\node (p1) at (0, 3*\yOffset) {$p_1$};
\node (p2) at (0, 2*\yOffset) {$p_2$};
\node (p3) at (0, \yOffset) {$p_3$};
\node (time) at (.2, .05) {\textbf{time}};

\node (time) at (4.5, -.2) {\fontSize linearization points};

\draw[->] (.2, 3*\yOffset) -- (\xLength, 3*\yOffset);
\draw[->] (.2, 2*\yOffset) -- (\xLength, 2*\yOffset);
\draw[->] (.2, \yOffset) -- (\xLength, \yOffset);
\draw[->,thick] (.6, 0) -- (\xLength, 0);

\node (prop1) at (2, 3*\yOffset+.2) {\fontSize $\propose(1)/\{1,2\}$};
\node (prop2) at (3.2, 2*\yOffset+.2) {\fontSize $\propose(2)/\{1,2\}$};
\node (prop3) at (4.5, \yOffset+.2) {\fontSize $\propose(3)/\{1,2,3\}$};

\draw[latex-latex,line width=.05cm,blue] (.7, 3*\yOffset) -- (2.5, 3*\yOffset);
\draw[latex-latex,line width=.05cm,blue] (1.3, 2*\yOffset) -- (3, 2*\yOffset);
\draw[latex-latex,line width=.05cm,blue] (3.5, \yOffset) -- (5.5, \yOffset);

\node[circle,fill=purple,minimum size=\circleR cm,inner sep=0pt] (lin1p1) at (\xLin1, 3*\yOffset) {};
\node[circle,fill=purple,minimum size=\circleR cm,inner sep=0pt] (lin1p2) at (\xLin1, 2*\yOffset) {};
\node[circle,fill=purple,minimum size=\circleR cm,inner sep=0pt] (lin1t) at (\xLin1, 0) {};
\draw[-latex,line width=.05cm,purple] (\xLin1, 3*\yOffset) -- (\xLin1, .05);

\node[circle,fill=purple,minimum size=\circleR cm,inner sep=0pt] (lin2p3) at (\xLin2, \yOffset) {};
\node[circle,fill=purple,minimum size=\circleR cm,inner sep=0pt] (lin2t) at (\xLin2, 0) {};
\draw[-latex,line width=.05cm,purple] (\xLin2, \yOffset) -- (\xLin2, .05);

\end{tikzpicture}
    \caption{A set-linearizable execution of lattice agreement that is not linearizable.}
    \label{fig:set-linearizability}
\end{minipage}%
\hfil
\hspace{.2em}
\begin{minipage}{.48\textwidth}
    \centering
    \begin{tikzpicture}
    
\tikzmath{
\xLength=6;
\yOffset=.6;
\circleR=.2;
\xLin1 = 1.8;
\xLin2 = 3.5;
}

\newcommand{\fontSize}{\footnotesize}

\node (p1) at (0, 3*\yOffset) {$p_1$};
\node (p2) at (0, 2*\yOffset) {$p_2$};
\node (p3) at (0, \yOffset) {$p_3$};
\node (time) at (.2, .05) {\textbf{time}};

\node (time) at (4.5, -.2) {\fontSize linearization points};

\draw[->] (.2, 3*\yOffset) -- (\xLength, 3*\yOffset);
\draw[->] (.2, 2*\yOffset) -- (\xLength, 2*\yOffset);
\draw[->] (.2, \yOffset) -- (\xLength, \yOffset);
\draw[->,thick] (.6, 0) -- (\xLength, 0);

\node (prop1) at (2, 3*\yOffset+.2) {\fontSize $\propose(1)/\{1,2\}$};
\node (prop2) at (3.2, 2*\yOffset+.2) {\fontSize $\propose(2)/\{1,2,3\}$};
\node (prop3) at (4.8, \yOffset+.2) {\fontSize $\propose(3)/\{1,2,3\}$};

\draw[latex-latex,line width=.05cm,blue] (.7, 3*\yOffset) -- (2.5, 3*\yOffset);
\draw[latex-latex,line width=.05cm,blue] (1.3, 2*\yOffset) -- (4, 2*\yOffset);
\draw[latex-latex,line width=.05cm,blue] (3, \yOffset) -- (5.2, \yOffset);

\node[circle,fill=purple,minimum size=\circleR cm,inner sep=0pt] (lin1p1) at (\xLin1, 3*\yOffset) {};
\node[circle,fill=purple,minimum size=\circleR cm,inner sep=0pt] (lin1p2) at (\xLin1, 2*\yOffset) {};
\node[circle,fill=purple,minimum size=\circleR cm,inner sep=0pt] (lin1t) at (\xLin1, 0) {};
\draw[-latex,line width=.05cm,purple] (\xLin1, 3*\yOffset) -- (\xLin1, .05);

\node[circle,fill=purple,minimum size=\circleR cm,inner sep=0pt] (lin2p2) at (\xLin2, 2*\yOffset) {};
\node[circle,fill=purple,minimum size=\circleR cm,inner sep=0pt] (lin2p3) at (\xLin2, \yOffset) {};
\node[circle,fill=purple,minimum size=\circleR cm,inner sep=0pt] (lin2t) at (\xLin2, 0) {};
\draw[-latex,line width=.05cm,purple] (\xLin2, 2*\yOffset) -- (\xLin2, .05);

\end{tikzpicture}
    \caption{An interval-linearizable execution of lattice agreement that is not set-linearizable.}
    \label{fig:int-linearizability}
\end{minipage}
\vspace{-1em}
\end{figure}

To illustrate the set- and interval-linearizability consistency conditions, we provide some examples of executions of lattice agreement in \Cref{fig:set-linearizability,fig:int-linearizability}, taken from \cite{CRR23}.
Lattice agreement is an object that provides a single operation $\propose(v)/V$, where $v$ is a value and $V$ is a set of proposed values.
Its only safety property is that $V$ must contain all previously or concomitantly proposed values along with the value being proposed, and its only liveness property is that the \propose operation must eventually terminate for correct processes.

In the set-linearizability example of \Cref{fig:set-linearizability}, op-exes form two equivalence classes $\{\propose(1), \propose(2)\}$ and $\{\propose(3)\}$.
The last clause of \OpSetOrder enables the creation of said equivalence classes.
Indeed, we have $\propose(1) \oo \propose(2) \oo \propose(3)$ and $\propose(1) \oo \propose(3)$.
Besides, we also have $\propose(2) \oo \propose(1) \oo \propose(3)$ and $\propose(2) \oo \propose(3)$.
This shows that the forbidden pattern in set-linearizability is, for any op-exes $o,o',o''$ such that $o \neq o''$, there is $o \oo o' \oo o''$ and $o'' \not\!\!\!\oo o' \not\!\!\!\oo o$.
Hence, the weakened transitivity clause of \OpSetOrder precludes this pattern.
Note that the $o \neq o''$ condition in this clause prevents the contradiction of this clause with the irreflexivity property.

In the interval-linearizability example of \Cref{fig:int-linearizability}, equivalence classes can be more complex.
More precisely, two equivalence classes can intersect, but it does not necessarily imply that both equivalence classes can ``see'' each other.
Here, op-exes form two different equivalence classes $\{\propose(1), \propose(2)\}$ and $\{\propose(2), \propose(3)\}$.
This shows that the forbidden pattern in interval-linearizability is: for any op-exes $o,o',o''$ that are connected but not concurrent, \ie $(o \oo o' \oo o'' \land o'' \not\!\!\!\oo o' \not\!\!\!\oo o)$, we also have $o'' \oo o$.
The clause precluding this pattern is thereby $(o \oo o' \oo o'' \land o'' \not\!\!\!\oo o' \not\!\!\!\oo o) \implies o'' \not\!\!\!\oo o$.
However, because of the connectedness property, the $o \oo o' \oo o''$ part of the implication is redundant, and the formula can be simplified to $o'' \not\!\!\!\oo o' \not\!\!\!\oo o \implies o'' \not\!\!\!\oo o$.
Finally, by applying the contrapositive, we obtain the formulation of the clause that appears in \OpIntOrder: $o \oo o'' \implies (o \oo o' \lor o' \oo o'')$.



\section{History Correctness}
\label{sec:histories}


Thus far, we have defined \textit{Legality} (\Cref{sec:objspec}) and extended it to the \textit{Consistency} (Section \ref{sec:consistency}) of a history $H$ with respect to an op-ex order \oo, as a set of clauses \cond. 
We now define the correctness of a history $H$ with respect to a set of clauses \cond when no process is Byzantine.

\begin{definition}[Correctness predicate]
\label{def:Correct}
Given a history $H=(E,\eo,O,P)$ and a set of clauses \cond, the following predicate describes the correctness of $H$ with respect to \cond:
$$
\Correctness(H, \cond) \triangleq \Exists \oo \in O^2: \bigwedge\limits_{C \in \cond} C(H,\oo).
$$
\end{definition}

Intuitively, a history $H$ is correct with respect to a set of clauses \cond if it is possible to find a relation \oo on the op-exes of $H$, such that all clauses in \cond are satisfied.
As an example, $\Correctness(H,\ProcessConsistency)$ is the predicate that decides whether history $H$ is correct under \textit{ProcessConsistency}, which according to its definition in Section \ref{sec:consistency}, requires that the clauses composing \textit{Legality} (see Section \ref{sec:objspec}) and \textit{OpProcessOrder} (see Section \ref{sec:consistency}) are satisfied by $H$. 
\nn{Note by the above definition of correctness, it is apparent that the more clauses are present, the fewer histories, and thus executions, will satisfy all the clauses. This demonstrates that when stronger, more restrictive, semantics are considered, the more refined is the set of executions that can provide them.}

In a similar fashion, we can derive a more general definition where processes may exhibit Byzantine behavior. 
To model the set of all possible Byzantine behaviors, we introduce the \byzhist function, which, given a history $H$, returns the set of all modified histories $H'$, where the op-exes by non-Byzantine (\ie \correct or \omitting) processes are the same in $H$ and $H'$, but Byzantine processes are given any arbitrary set of pending op-exes.

\begin{definition}[Byzantine histories function]
\sloppy
Given history $H =(E,\!\eo,\!O,P)$, the $\byzhist(H)$ function returns the set of all possible histories $H'=(E',\eo',O',P)$ s.t.
\begin{align*}
    O' &= \{o \in O \mid o.\proc.\type \neq \faulty.\byzantine\} \cup \{\text{any arbitrary set of pending} \\
    &\hspace{1.7em}\text{op-exes by $p$} \mid \Forall p \in P, p.\type=\faulty.\byzantine\}, \\
    E' &= \{i,r \in (i,r) \in O'\} \text{, and } \eo' \subseteq E'^2: \eo \subseteq \eo'.
\end{align*}
\end{definition}

\sloppy{Informally, given a base history $H=(E,\eo,O,P)$ and a modified history $H'=(E',\eo',O',P) \in \byzhist(H)$, the set $O'$ is constructed by keeping all op-exes of $O$ by non-Byzantine processes and creating arbitrary pending op-exes for Byzantine processes, the set $E'$ is the set of all events appearing in $O'$, and the order $\eo'$ is an arbitrary total order on $E'$ extending \eo.}
Notice that we only populate the op-exes of Byzantine processes using pending op-exes, and not complete op-exes or notifications, as we do not guarantee anything for Byzantine processes.
Hence, we define correctness with Byzantine processes as follows.

\begin{definition}[Byzantine Correctness predicate]
\label{def:byz-correct}
Given a history $H=(E,\eo,O,P)$ and a set of clauses \cond, the following predicate describes the \textit{Byzantine} correctness of $H$: 
$$
\Correctness(H, \cond) \triangleq \Exists H'=(E',\eo',O',P) \in \byzhist(H), \Exists \oo \in O'^2: \bigwedge\limits_{C \in \cond} C(H',\oo).
$$
\end{definition}





Intuitively, a history $H$ with Byzantine processes is correct with respect to a set of clauses \cond if it is possible to construct a modified history $H'$ (where Byzantine processes perform arbitrary op-exes) and an arbitrary relation \oo on the op-exes of $H'$, such that all clauses in \cond are satisfied. 
To create the set of all possible modified histories, we use the \byzhist function.
\af{In other words, history $H$ is correct if and only if we can ``fix'' it by changing only the op-exes of the Byzantine processes to make it correct with respect to $\cond$.}
In the absence of Byzantine processes, \Cref{def:byz-correct} collapses to \Cref{def:Correct}.


\section{Examples of Object Specifications} \label{sec:obj-spec-examples}
To exemplify the usability of our formalism, we specify reliable broadcast, shared memory, and message passing as concurrent objects. The modularity of our
formalism is further demonstrated by combining the consistency conditions of Section~\ref{sec:consistency} with the broadcast and share memory object  
specifications, obtaining broadcast and shared memory hierarchies.


\subsection{Reliable Broadcast Object}
\label{sec:examples}

We begin the object specification examples by formally defining the celebrated \emph{reliable broadcast} problem~\cite{B87}.
%

Let us remark that our formalism allows us to create object specifications and consistency hierarchies that are completely independent of the failure model: they hold both for omission (\eg crashes) and Byzantine faults.
Let us also observe that our framework's modularity enables us to define various consistent object specifications effortlessly by simply combining an object definition with a consistency condition.
For example, by combining \Linearizability with reliable broadcast (as specified below), we obtain another abstraction, {linearizable broadcast~\cite{CK21}.

In the following, the specifications consist of a list of operations with their correctness predicates, \validp, \safep, and \livep.
For concision, if we do not explicitly specify the \validp, \safep or \livep predicates for some operation, then it means that implicitly, these predicates always evaluate to \ttrue.
Furthermore, we use in the following logical formulas the $\gets$ symbol to denote an \JW{affectation ->}\jw{assignment} of a value to a variable in the predicates.
For convenience, we define below a shorthand for referring to the set of correct processes.

\begin{definition}[Set of correct processes]
For a given history $H=(E,\eo,O,P)$, we define the function $\corr(H)$ that returns the set of correct processes of $H$, \ie
$\corr(H) \triangleq \{p \in P \mid p.\type = \correct\}.$
\end{definition}

Below, we define a reusable specification property for liveness checking that, if the process of an op-ex is correct, then this op-ex must terminate.

\begin{definition}[Op-ex termination]
\label{def:termination}
For an op-ex $o$, the \emph{op-ex termination} liveness property is defined as
$
    \OpTermination(o) \triangleq o.\proc.\type = \correct \implies o \not\equiv \opp()/\bot.
$
\end{definition}



\subsubsection{Reliable Broadcast Specification} \label{sec:rbcast}
\emph{Reliable broadcast} is a fundamental abstraction of distributed computing guaranteeing an \textit{all-or-nothing} delivery of a message that a sender has broadcast to all processes of the system, and this despite the potential presence of faults (crashes or Byzantine)~\cite{B87}.
This section considers the multi-sender and multi-shot variant of reliable broadcast, where every process can broadcast multiple messages (different messages from the same process are differentiated by their message ID). 
A \emph{reliable broadcast} object $B$ provides the following operations:
\begin{itemize}
    \item $B.\rbroadcast(m,\id)$: broadcasts message $m$ with ID \id,
    
    \item $B.\rdeliver/(m,\id,i)$ (notification): delivers message $m$ with ID \id from process $p_i$.
\end{itemize}

In the following, we consider a multi-shot reliable broadcast object $B$, a set of op-exes $O$, a relation \oo on $O$, an op-ex $o \in O$ and its context $(O_c,\oo_c) = \ctx(o,O,\oo)$.

\paragraph{Operation \rbroadcast.} If $o \equiv B.\rbroadcast_i(m,\id)$, then we have the following.
\begin{align*}
    \rbroadcast.\validp(o,(O_c,\oo_c)) \!&\triangleq\! \Nexists \rbroadcast_i(-,\id) \in O_c. \\
    \rbroadcast.\livep(o,H,\oo) \!&\triangleq\! \OpTermination(o) \\
    &\hspace{-1em}\land (\Forall p_j \in \corr(H), \Exists o' \equiv B.\rdeliver_j/(m,\id,i) \in O: o \oo o').
\end{align*}

The \validp predicate states that a process cannot broadcast more than once with a given ID.
The \livep predicate states that a \rbroadcast op-ex must terminate if a correct process made it, and must trigger matching \rdeliver op-ex on every correct process.

\paragraph{Operation \rdeliver.} If $o \equiv B.\rdeliver_i/(m,\id,j)$, then we have the following.
\begin{align*}
    \rdeliver.\safep(o,(O_c,\oo_c)) &\triangleq \\
    C \gets \{B.\rbroadcast&_j(m',\id') \in O_c \mid \Nexists B.\rdeliver_i/(m',\id',j) \in O_c\}, \\
    \hspace{1.4em}F \gets \{(m',\id') \mid \Forall b,&b' \in C, b \equiv B.\rbroadcast_j(m',\id'): b' \not\!\!\oo_c b\}: 
    (m,\id) \in F. \\
    \rdeliver.\livep(o,H,\oo) &\triangleq p_i \in \corr(H) \\
    &\implies (\Forall p_j \in \corr(H), \Exists B.\rdeliver_j/(m,\id,k) \in O).
\end{align*}

The \safep predicate states that a delivery must return one of the first broadcasts that have not been delivered with respect to $\oo_c$.
In \safep, $C$ denotes the set of candidate broadcast op-exes that have not been delivered, and $F$ denotes the set of ``first'' broadcast values (message and ID) of op-exes of $C$ that are not preceded (w.r.t. $\oo_c$) by other op-exes in $C$.
Notice that this does not necessarily mean that broadcasts must be delivered in FIFO order, as \oo does not necessarily follow FIFO order (to have this property, \oo would have to follow \FIFOConsistency, see \Cref{sec:consistency}).
This is the reverse of registers, where you can only read one of the last written values according to \oo.
The \livep predicate states that, if a correct process delivers a message, then all correct processes deliver this message.

\subsubsection{The Reliable Broadcast Consistency Hierarchy} \label{sec:obj-cons-examples}
The modularity of our formalism allows us to plug any consistency condition (\eg the ones defined in \Cref{sec:cons-cond}), or set of consistency conditions, that we want on any given object specification (\eg reliable broadcast)
to yield a \textit{consistent object specification}.
This section demonstrates this fact by applying different consistency conditions on the previously defined reliable broadcast specification.

A reliable broadcast object can provide different ordering guarantees depending on which consistency conditions it is instantiated with.
\Cref{fig:bcast-hierarchy} illustrates the reliable broadcast hierarchy, and how reliable broadcasts of different strengths can be obtained by using \ProcessConsistency, \FIFOConsistency, \CausalConsistency, \Serializability or \Linearizability.

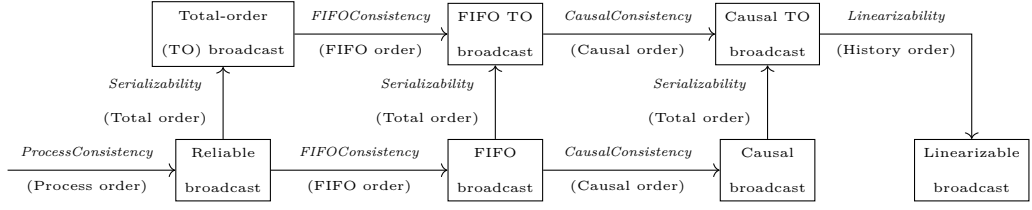
\begin{figure}[t]
    \centering
    \begin{tikzpicture}
\tikzmath{
\yOffset=1.8;
\wNode=0;
\xFIFO=3.6;
\xCausal=7.2;
\xLin=9.9;
\xMid1=(\xCausal+\xFIFO)/2-.05;
\xMid2=\xLin-1;
}

\newcommand{\fontSize}{\tiny}

\node[draw,align=center,minimum width=\wNode cm] (rb) at (0,0) {\fontSize Reliable\\\fontSize broadcast};

\node[draw,align=center,minimum width=\wNode cm] (fb) at (\xFIFO, 0) {\fontSize FIFO\\\fontSize broadcast};

\node[draw,align=center,minimum width=\wNode cm] (cb) at (\xCausal, 0) {\fontSize Causal\\\fontSize broadcast};

\node[draw,align=center,minimum width=\wNode cm] (tb) at (0, \yOffset) {\fontSize Total-order\\\fontSize (TO) broadcast};

\node[draw,align=center,minimum width=\wNode cm] (ftb) at (\xFIFO, \yOffset) {\fontSize FIFO TO\\\fontSize broadcast};

\node[draw,align=center,minimum width=\wNode cm] (ctb) at (\xCausal, \yOffset) {\fontSize Causal TO\\\fontSize broadcast};

\node[draw,align=center,minimum width=\wNode cm] (lb) at (\xLin, 0) {\fontSize Linearizable\\\fontSize broadcast};


\draw[->] (-2.85,0) -- (rb);
\node[align=center] at (-1.8,0) {\fontSize\ProcessConsistency\\\fontSize(Process order)};

\draw[->] (rb) -- (fb);
\node[align=center] at (\xFIFO/2, 0) {\fontSize\FIFOConsistency\\\fontSize(FIFO order)};
\draw[->] (fb) -- (cb);
\node[align=center] at (\xMid1, 0) {\fontSize\CausalConsistency\\\fontSize(Causal order)};

\draw[->] (tb) -- (ftb);
\node[align=center] at (\xFIFO/2+.15, \yOffset) {\fontSize\FIFOConsistency\\\fontSize(FIFO order)};
\draw[->] (ftb) -- (ctb);
\node[align=center] at (\xMid1, \yOffset) {\fontSize\CausalConsistency\\\fontSize(Causal order)};

\draw[->] (rb) -- (tb);
\node[align=center] at (-.9, \yOffset/2) {\fontSize\Serializability\\\fontSize(Total order)};
\draw[->] (fb) -- (ftb);
\node[align=center] at (\xFIFO-.9, \yOffset/2) {\fontSize\Serializability\\\fontSize(Total order)};
\draw[->] (cb) -- (ctb);
\node[align=center] at (\xCausal-.9, \yOffset/2) {\fontSize\Serializability\\\fontSize(Total order)};

\draw[->] (ctb) -- (\xLin, \yOffset) -- (lb);
\node[align=center] at (\xMid2, \yOffset) {\fontSize\Linearizability\\\fontSize(History order)};

\end{tikzpicture}
    \caption{The reliable broadcast hierarchy of~\cite{HT94,R18} extended with linearizable broadcast~\cite{CK21}.}
    \label{fig:bcast-hierarchy}
\end{figure}


As we can see, to obtain simple reliable broadcast, we must use the \ProcessConsistency condition to guarantee that the op-exes of a given process are totally ordered.
This assumption is necessary for the invocation validity (the precondition) of the \rbroadcast operation, defined by the $\rbroadcast.\validp$ predicate.
Indeed, this predicate states that a process cannot broadcast twice with the same ID; however, if op-exes of a process are not totally ordered, then there can be two \rbroadcast op-exes from the same process and with the same ID that would not be in the context of one another, and thus the $\rbroadcast.\validp$ would not be violated when it should be.
This is why a per-process total order of op-exes (imposed by \ProcessConsistency) is often required for some object specifications (and in this case, for reliable broadcast).



\subsection{Shared Memory Object} 
\label{sec:shared-mem}

\emph{Shared memory} is a communication model where system processes communicate by reading and writing on an array of registers, identified by their address. We proceed with its formal specification. 

\subsubsection{Shared Memory Object Specification}

A shared memory $M$ provides the following operations:
\begin{itemize}
    \item $M.\rread(a)/v$: returns one of the latest values $v$ written in $M$ at address $a$,
    
    \item $M.\wwrite(v,a)$: writes value $v$ in $M$ at address $a$.
\end{itemize}

In the following, we consider a shared memory $M$, a set of op-exes $O$, a relation \oo on $O$, an op-ex $o \in O$ and its context $(O_c,\oo_c) = \ctx(o,O,\oo)$.

\paragraph{Operation \rread.} If $o \equiv M.\rread_i(a)/v$, then we have the following.
\begin{align*}
    \rread.\validp(o,(O_c,\oo_c)) &\triangleq \Exists o' \equiv M.\wwrite(-,a) \in O_c. \\
    \rread.\safep(o,(O_c,\oo_c)) &\triangleq v \in \{v' \mid \Exists o' \equiv M.\wwrite(v',a) \in O_c, \Nexists o'' \equiv M.\wwrite(-,a) \in O_c':
    \\ &\hspace{3.6em}
    o' \oo_c o''\}. \\
    \rread.\livep(o,H,\oo) &\triangleq \OpTermination(o).
\end{align*}

The \validp predicate states that a process cannot read an address never written into.
The \safep predicate states that a read must return one of the last written values at that address with respect to $\oo_c$.
The \livep predicate states that a \rread op-ex must terminate if a correct process made it.

\paragraph{Operation \wwrite.} If $o \equiv M.\wwrite_i(v,a)$, then we have the following.
\begin{align*}
    \wwrite.\livep(o,H,\oo) &\triangleq \OpTermination(o).
\end{align*}
The \livep predicate states that a \wwrite op-ex must terminate if a correct process made it.

\subsubsection{Possible Variants} 

In the above, we have defined a version of shared memory constituted of multi-writer multi-reader registers (abridged MWMR), where everyone can read and write all the registers.
But if we want to restrict the access of some registers to some processes, we can use the \validp precondition of the \rread and \wwrite operations.
For example, if we want to design a single-writer multi-reader register (abridged SWMR), we can impose in the $\wwrite.\validp$ predicate that only the invocations of \wwrite by a single process are considered valid.
More generally, we can design asymmetric objects that provide different operations to different system processes using this technique.

\subsubsection{The Shared Memory Hierarchy} 

As illustrated by \Cref{fig:mem-hierarchy}, by applying the \FIFOConsistency, \CausalConsistency, \SequentialConsistency or \Linearizability consistency conditions on the specification of shared memory, 
different kinds of memory consistencies can be obtained.
\begin{figure}[ht]
    \centering
    \begin{tikzpicture}
\tikzmath{
\xCaus=4.4;
\xSeq=8.4;
\xAtom=12.3;
\yOffset=2;
\wNode=1.5;
}

\newcommand{\fontSize}{\small}


\node[draw,align=center,minimum width=\wNode cm] (pram) at (0,0) {\fontSize PRAM\\\fontSize memory};

\node[draw,align=center,minimum width=\wNode cm] (caus) at (\xCaus,0) {\fontSize Causal\\\fontSize memory};

\node[draw,align=center,minimum width=\wNode cm] (seq) at (\xSeq,0) {\fontSize Sequential\\\fontSize memory};

\node[draw,align=center,minimum width=\wNode cm] (atom) at (\xAtom,0) {\fontSize Atomic\\\fontSize memory};


\draw[->] (0,-1) -- (pram);
\node[align=center] at (1.35,-.8) {\fontSize\FIFOConsistency};

\draw[->] (pram) -- (caus);
\node[align=center] at (\xCaus/2-.05,0) {\fontSize\CausalConsistency\\};

\draw[->] (caus) -- (seq);
\node[align=center] at (\xSeq/2+\xCaus/2-.05,0) {\fontSize\SequentialConsistency\\};

\draw[->] (seq) -- (atom);
\node[align=center] at (\xAtom/2+\xSeq/2,0) {\fontSize\Linearizability\\};
\end{tikzpicture}\vspace{-1em}
    \caption{The shared memory hierarchy.}
    \label{fig:mem-hierarchy}
\end{figure}
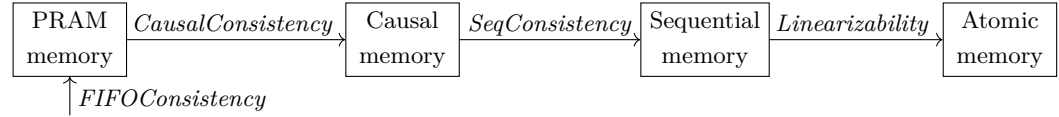

\subsection{Asynchronous Message-passing Object} \label{sec:async-msg-pass}

\emph{Asynchronous message-passing} is a communication model where system processes communicate by sending and receiving messages.
This model is said to be asynchronous because messages can have arbitrary delays. We proceed with its formal specification.

\subsubsection{Asynchronous Message-passing Object Specification}

A message-passing object $M$ provides the following operations:
\begin{itemize}
    \item $M.\send(m,i)$: sends message to receiver $p_i$,
    
    \item $M.\receive/(m,i)$ (notification): receives message $m$ from process $p_i$.
\end{itemize}

In the following, we consider an asynchronous message-passing object $M$, a set of op-exes $O$, a relation \oo on $O$, an op-ex $o \in O$ and its context $(O_c,\oo_c) = \ctx(o,O,\oo)$.

\paragraph{Operation \send.} If $o \equiv M.\send_i(m,j)$, then we have the following.
\begin{align*}
    \send.\livep(o,H,\oo) &\triangleq \OpTermination(o) \\
    &\land (p_j \in \corr(H), \Exists o' \equiv M.\receive_j/(m,i) \in O: o \oo o').
\end{align*}

The \livep predicate states that a \send op-ex must terminate if a correct process made it, and that the receiver, if it is correct, must eventually receive the message.
For simplicity, we assume that a given message is only sent once (so we do not have to guarantee that it is received as often as it has been sent.

\paragraph{Operation \receive.} If $o \equiv M.\receive_i/(m,j)$, then we have the following.
\begin{align*}
    \receive.\safep(o,(O_c,\oo_c)) \triangleq (m,j) \in \{&(m',k) \mid \\
    &\Exists M.\send_k(m',i) \in O_c, \Nexists M.\receive_i/(m',k) \in O_c\}.
\end{align*}
The \safep predicate states that if a process receives a message, then this message has been sent before.

\subsubsection{Possible Variants} 

We considered in this specification the asynchronous message-passing model, in which messages have arbitrary delays.
But let us mention that this model's \emph{synchronous} counterpart, where messages have a maximum delay known by all processes, can also be represented in our formalism as a concurrent object.
The synchronous message-passing model can be represented as having rounds of communication, where all the messages sent in a round are received in the same round.
Hence, we see that a synchronous message-passing object $S$ can be represented as providing two operations $S.\send(m,i)$ and \jw{$S.\Endround/M$}, where \Endround is a notification delivering to the process at hand $p_i$ all the set $M$ of messages sent to $p_i$ during the round that ended.
Again, let us notice that our formalism can specify the behavior of complex distributed systems without relying on higher-order logic such as temporal logic.

Furthermore, we assumed a message-passing specification over reliable channels; that is, there is no message corruption, deletion, duplication, \etc, for instance, due to interference or disconnections.
We classify this kind of network failure under the message adversary model~\cite{SW89}.
However, we can easily imagine variants of this specification that consider a message adversary.
In particular, for message deletions, the techniques introduced in~\cite{AFRT23} can help us to design a message-adversary-prone asynchronous message-passing object.

Finally, we considered an authenticated message-passing object because, when a message is received, the recipient knows the sender's identity (there is no identity spoofing), but we can easily design an unauthenticated variant that does not provide this information.

\section{Impossibility of Resilient Consensus in Asynchronous
Systems}
\label{sec:consensus}

This section further exemplifies the framework's utility by showing how it can be used to construct axiomatic proofs.
Particularly, we use our framework to define a Consensus object (Section~\ref{sec:codef}) and to provide an axiomatic proof (Sections~\ref{sec:asynconcsystem}--\ref{sec:impTheorem}) of the FLP impossibility of having reliable deterministic consensus in an asynchronous system with process failures~\cite{FLP85}. (This proof is inspired by \cite{T91}.) 

Note that our proof is agnostic of the communication medium used by the processes to communicate. For completeness,
we show in Section~\ref{sec:asych-axiom-register-channel} that SWSR atomic registers and point-to-point message-passing
channels satisfy the relevant assumptions of the proof.

\subsection{Consensus Object}
\label{sec:codef}
We start by providing the specification of a Consensus object using the conditions defined in \Cref{sec:objspec}.
Our Consensus object $\Consensus$ has only one notification operation, $\Consensus.\decide_i(v)$, which returns a value $v \in V$ (we have binary consensus if $V = \{0,1\}$) to process $p_i$. Observe that we consider a simple version of a Consensus object without the common $\propose$ operation. Proving the impossibility of this version makes our proof more general\footnote{We could add a $\Consensus.\propose_i(v)$ operation to the Consensus object that returns nothing. The validity predicate of the $\decide()$ notification has to be adapted accordingly, but this does not affect the proof.}.

Let $\Histories$ be the set of histories of a distributed system that contains a Consensus object $\Consensus$.
For every history $H=(E,\eo,O, P) \in \Histories$, let $H|\Consensus=(E|C,{\eo|C},O|C, P)$ be the subhistory containing only the events of $E$ applied to $\Consensus$.
Consider history $H=(E,\eo,O, P) \in \Histories$ with set of op-exes $O$, a relation \oo on $O$, an op-ex $o \in O|\Consensus$, and its context $\ctx(o,O,\oo)=(O_c,\oo_c)$.

\paragraph{Operation \decide.} If $o \equiv \Consensus.\decide_i/v$, then we have the following predicates.
\begin{align*}
    \decide.\safep(o,(O_c,\oo_c)) &\triangleq (v \in V)
    \land (\Forall \Consensus.\decide_j/v' \in O_c: v=v'). \\
    \decide.\livep(o,H,\oo) &\triangleq \Exists \Consensus.\decide_j/- \in O.
\end{align*}
The \safep predicate states that the values decided are in the appropriate set $V$ and that, in the context of each op-ex, 
all decided values are the same.
Observe that we allow the same process to decide several times as long as the decided values are the same.
The \livep predicate states that some process must decide in every history.

The Consensus object must guarantee that exactly one value can be decided in each history.
We achieve this by combining the Consensus object specification with the $\Serializability$ consistency.

\begin{assumption}
\label{assum:serializability}
    For every history $H=(E,\eo,O,P) \in \Histories$, we must have $\Correctness(H|\Consensus, \Serializability)$.
\end{assumption}

Observe that $\Serializability$ is only imposed on $H|\Consensus$, that is, we only impose a total order on the \decide op-exes.
From the fact that ${\oo|C}$ is a total order of the op-exes on object $\Consensus$ (imposed by $\Serializability$), and the last clause of $\decide.\safep()$, all \decide op-exes in $O$ return the same value $v$.

\begin{observation}
\label{obs:one-value}
$\Forall H=(E,\eo,O,P) \in \Histories,  (\Consensus.\decide_i/v_i,\Consensus.\decide_j/v_j \in O) \implies (v_i=v_j)$.
\end{observation}

Observe that it is possible to have trivial implementations of a Consensus object in which all histories decide the same hardcoded value $v \in V$.
Unfortunately, this object is not very useful.
We will impose below a non-triviality condition that guarantees that there are histories in which the Consensus object decides different values.
Additionally, the set of $\Histories$ must reflect the fact that the system is asynchronous and has $n$ processes of which up to one can crash. 

\subsection{Asynchronous Distributed System}
\label{sec:asynconcsystem}

We consider an asynchronous distributed system with $n$ processes in which up to one process can crash. This means that in any history $H=(E,\eo,O,P) \in \Histories$ of the system $|P|=n$ and at most one process $p \in P$ has $p.\type=\faulty$. For convenience we assume that $P=\{p_1, p_2, \ldots, p_n\}$ in all histories.

The set of objects \Objects contains a crash-resilient Consensus object $\Consensus$. In order to be able to solve consensus, it also contains some object $M$ that allows processes to communicate. Observe that the events of this communication medium $M$ are in the system histories.

The (potentially infinite) set $\Histories$ represents the system executions.
From these histories, we will construct a (potentially infinite) set $\States$ of possible states of the system.
Each \emph{state} $\state \in \States$ is a (potentially infinite) set of events.
Intuitively, a state $\state$ is the collection of local states of all system processes $p_i$, represented by the totally-ordered local events that $p_i$ has experienced.

To define the set $\States$, we first assign an \emph{index} to each event in a history $H=(E,\eo,O,P) \in \Histories$.
The index assigned to event $e \in E$ is an attribute $e.\idx$ that is the position of $e$ in the sequence of events of its process $e.\proc$.
This sequence is obtained by ordering with $\eo$ the set $E|e.\proc$.
Observe that the sets of events of different histories may have common events.
After adding the indices, common events with the same index are the same, but with different indices are different.
For instance, the indices distinguish events in two histories in which the same process receives the same messages from the same senders but in different orders.

A special subset of $\States$ is the set of \emph{complete states}, defined as 
    $\complete(\States) \triangleq \{E \mid (E,\eo,O,P) \in \Histories \}.$
Consider now any history $H=(E,\eo,O,P) \in \Histories$.
We say that state $\state = E \in \complete(\States) \subseteq \States$ is a \emph{state extracted from $H$.}
Then, we apply iteratively and exhaustively the following procedure to add more states to $\States$: if $\state \in \States$ is a state extracted from $H$, let $e$ be the event in $\state$ with the largest index $e.\idx$ of those from process $e.\proc$,
then $\state \setminus \{e\}$ is also a state in $\States$ extracted from $H$.
This procedure ends when the empty state $\state=\varnothing$ is reached (which is also in $\States$).
Hence, $\States$ contains a state $\state$ iff there is a history $H=(E,\eo,O,P) \in \Histories$ such that the events in $\state$ from every process $p \in \Processes$ are a prefix of the sequence of all events from $p$ in $E$ ordered by $\eo$.

Observe that this construction of the set $\States$ guarantees the following property:
$$
    \Continuity(\States) \triangleq \Forall \state \in \States \setminus \varnothing, \Exists e \in \state: (\state \setminus \{e\} \in \States).
$$

Moreover, the set of states $\States$ of the asynchronous system must satisfy the following axiom.

\begin{definition}[Asynchronous distributed system axiom] \label{def:asynch-short}
The following predicate holds for the set of states $\States$ in an asynchronous distributed system.
\begin{align*}
\Asynchrony(\States) \triangleq &\Forall \state \!\in\! \States, (\state \cup \{e\} \!\in \States \land \state \cup \{e'\} \!\in \States \land e.\proc \neq e'.\proc) 
\\&
\!\!\implies\!\! (\state \cup \{e, e'\} \!\in \States).
\end{align*}
\end{definition}

\Asynchrony requires that if two states differ only in their last respective events, which are from different processes, their union is also a state. Observe that this must hold even if the two states are extracted from different histories.
We point out that our impossibility proof is agnostic of the communication medium object $M$, as long as the medium satisfies asynchrony as defined above.

\subsection{Valence}

Our impossibility proof relies on the notion of valence, which was first introduced in~\cite{FLP85}.

\begin{definition}[Valence function \valence]
Given a state $\state \in \States$, the \emph{valence} of $\state$ is a set of values given by $\valence(\state)$ as follows.
\begin{itemize}
    \item If state $\state \in \complete(\States)$ and $\state$ is extracted from history $H=(E,\eo,O,P)$, then $\valence(\state)= \{ v \mid \Consensus.\decide()/v \in O \}$.
    \item
    $\Branching(\States)$: $\Forall \state \in \States \setminus \complete(\States), \valence(\state)=\bigcup_{\state' \in \{\state \cup \{e\} \in \States \mid e \notin \state\}} \valence(\state')$.
\end{itemize}
\end{definition}

Intuitively, the valence of a complete state is the set of all values that were decided in the histories from which it was extracted, and, by \Branching, the valence of an incomplete state is the union of the valences of all its one-event extensions.
We say that a state $\state \in \States$ is \emph{univalent} iff we have $|\valence(\state)|=1$, and we say that it is \emph{multivalent} iff we have $|\valence(\state)|>1$.
Observe that it is not possible that $|\valence(\state)|=0$, due to the \livep predicate of the Consensus object.


\begin{restatable}{lemma}{nonemptyvalence} \label{lem:nonemptyvalence}
$\PositiveValence(\States) \triangleq \Forall \state \in \States, |\valence(\state)| \geq 1$.
\end{restatable}

\begin{proof}
From the $\livep$ predicate (liveness) of the Consensus object specification, the valence $\valence(\state)$ of a complete state $\state \in \complete(\States)$ extracted from some history $H$ contains at least one value.
Let us consider now a state $\state \in \States \setminus \complete(\States)$ and assume by induction that all its one-event extension states $\state' \in \{\state \cup \{e\} \in \States \mid e \notin \state\}$ have $|\valence(\state')| \geq 1$.
By construction of the set of states $\States$, there is at least one such one-event extension state $\state'$. 
By $\Branching(\States)$, if holds that $\valence(\state') \subseteq \valence(\state)$. 
Then, $|\valence(\state)| \geq 1$.
\end{proof}

Moreover, it holds that all complete states have a finite univalent sub-state.

\begin{restatable}{lemma}{termination} \label{lem:termination}
$\Termination(\States) \triangleq \Forall \state \in \complete(\States), \Exists \state' \!\in\! \States, \state' \subseteq \state, |\state'| < +\infty: \valence(\state)=\valence(\state') \land |\valence(\state')|=1$.
\end{restatable}

\begin{proof}
Assume $\state \in \complete(\States)$ is extracted from history $H=(E,\eo,O,P)$.
First note that $|\valence(\state)|=1$ from \Cref{obs:one-value}.
Let $\valence(\state)=\{v\}$ and $\Consensus.\decide_i/v \in O$.
Then, in $\state$ there is a \decide event $e_d$ from process $p_i$ that returns $v$.
Then, $\state'=\{e \in \state \mid (e.\proc=p_i) \land (e \eo e_d) \} \cup \{e_d\}$ is finite and has $\valence(\state')=\{v\}$.
\end{proof}



\subsection{Resilient Non-trivial Consensus}
\label{sec:resilience}

Let us now define the properties we require for a non-trivial Consensus object that is resilient to any stopping process.

\begin{definition}[Resilient Non-trivial Consensus axioms] \label{def:cons}
Given a system $\States$, the following predicates describe \emph{resilient non-trivial consensus}.
\begin{align*}
    &\NonTriviality(\States) \triangleq \Exists \state,\state' \in \States: \valence(\state) \neq \valence(\state'). \\
    &\Resilience(\States) \triangleq \Forall \state \in \States, |\valence(\state)|>1, \Forall p \in P, 
    \Exists \state' {=} \state \cup \{e\} \in \States: (e \!\notin\! \state) \land (p \!\neq\! e.\proc).
\end{align*}
\end{definition}

\NonTriviality states that there exist 2 states with different valences, implying that there are histories deciding different values.
\Resilience states that, for any process, any multivalent state can be extended by an event that is not from this process.
This guarantees that even if one process stops taking steps (\ie crashes), the system can still progress and eventually reach a decision.

\subsection{Impossibility Theorem}
\label{sec:impTheorem}


\begin{theorem}
\label{thm:impossible}
There cannot be a resilient non-trivial Consensus object in an asynchronous system.
\end{theorem}

\begin{proof}
\sloppy
By way of contradiction, let us assume that we have a resilient non-trivial Consensus object $\Consensus$ in an asynchronous distributed system, and let $\States$ be the set of states obtained from the set $\Histories$ of the system as described above.
By construction, the properties \Continuity, \Branching, \PositiveValence, and \Termination hold for $\States$.
By assumption of an asynchronous distributed system, the axiom \Asynchrony of \Cref{def:asynch-short} holds.
We also assume that the axioms \NonTriviality and \Resilience of \Cref{def:cons} hold for $\States$, since $\Consensus$ is resilient and non-trivial. 

We first show that at least one multivalent state exists, \ie $\Exists \state \in \States: |\valence(\state)|>1$.
From \NonTriviality, we have two states $\state_u$ and $\state_{u'}$ with different valences.
From $\PositiveValence(\States)$, $\state_u$ and $\state_{u'}$ do not have empty valences, so they are either multivalent or univalent.
If either $\state_u$ or $\state_{u'}$ is multivalent, we are done, so let us assume that they are both univalent.
By \Continuity and \Branching, we can iteratively remove one event in these states, until we reach a state $\state_m$ ($\state_m$ can be the empty set) that is contained in both $\state_u$ and $\state_{u'}$ such that $\valence(\state_u) \cup \valence(\state_{u'}) \subseteq \valence(\state_m)$.
Hence, $\state_m$ is multivalent.

From the fact that there is some multivalent state $\state_m$, we can inductively show that there exists what we call a \emph{critical} state $\state_c$, \ie a multivalent state for which all extensions are univalent: 
$$
\Exists \state_c \in \States: (|\valence(\state_c)|>1) \land (\Forall \state' \in \States, \state_c \subset \state': |\valence(\state')|=1).
$$
Observe that $\state_m$ is incomplete (by $\Termination(\States)$) and hence has one-event extensions.
If all extensions are univalent, $\state_m$ satisfies the property of a critical state and we set $\state_c=\state_m$.
Otherwise, $\state_m$ has some one-event extension that is multivalent.
Then, we make $\state_m$ this new multivalent extension and repeat this procedure.
Observe that this process must eventually end by finding a critical state, since otherwise, it means an infinite multivalent state exists, which contradicts $\Termination(\States)$.

Let us remark that, given that $\state_c$ is a critical state, extending it by only one event results in a univalent state.
By \Branching, there exists (at least) two univalent states $\state_v, \state_{v'} \in \States$, with different valences and obtained extending $\state_c$ with one event: $\state_v = \state_c \cup \{e\}$ and $\state_{v'} = \state_c \cup \{e'\}$ such that $|\valence(\state_v)|=|\valence(\state_{v'})|=1$ and $\valence(\state_v) \neq \valence(\state_{v'})$.
Let us consider the two following cases.
\begin{itemize}
    \item Case 1: $e.\proc \neq e'.\proc$.
    Given that the processes of the two events are distinct, from \Asynchrony, we have $\state'=\state_v \cup \state_{v'} \in \States$.
    Since $\state' = \state_v \cup \{e'\}$, from \Branching it holds that $\valence(\state') \subseteq \valence(\state_v)$, and since $|\valence(\state')| \geq 1$ (\PositiveValence), then we have $\valence(\state')=\valence(\state_v)$.
    However, a similar argument yields that $\valence(\state')=\valence(\state_{v'})$, which contradicts $\valence(\state_v) \neq \valence(\state_{v'})$.
    
    \item Case 2: $e.\proc=e'.\proc$.
    By \Resilience, we can extend $\state_c$ with one event not from $e.\proc$ to get a state $\state''=\state_c \cup \{e''\} \in \States$, such that $e''.\proc \neq e.\proc$.
    From the criticality of $\state_c$, $\state''$ is univalent.
    Then, either $\valence(\state'') \neq \valence(\state_v)$ or $\valence(\state'') \neq \valence(\state_{v'})$.
    Without loss of generality, assume that $\valence(\state'') \neq \valence(\state_v)$.
    Then, the contradiction follows from Case~1.
    \qedhere
\end{itemize}
\end{proof}

\subsection{Asynchrony of SWSR Atomic Registers and Point-to-point Message Passing}
\label{sec:asych-axiom-register-channel}

We prove that SWSR atomic registers and message passing, as communication media, satisfy the \Asynchrony axiom of \Cref{def:asynch-short}. 
We make the following natural assumption about op-ex invocations.

\begin{assumption}[Process consistent behavior]
\label{asm:ProcessConsistenBehavior}
A process decides whether to invoke an op-ex based only on its local view. Formally,
\begin{align*}
    (\state \in \States  \land \state \cup \{e\} \in \States  \land e \equiv \opex.\inv) 
    \!\implies\!
    (\forall \state' \in \States: \state|e.\proc = \state'|e.\proc, \state' \cup \{e\} \in \States).
\end{align*}
\end{assumption}

\subsubsection{Asynchrony of an SWSR Atomic Register}

We prove that an SWSR atomic register satisfies asynchrony as defined in \Cref{def:asynch-short}.

\begin{theorem}
\label{thm:Async-atomicSWSR}
    A linearizable Single Writer Single Reader (SWSR) atomic register $R$ satisfies the asynchronous distributed system axiom of \Cref{def:asynch-short}.
\end{theorem}
\begin{proof}
Let us consider a system that contains a SWSR register $R$ as specified in \Cref{sec:objspec} with \Linearizability consistency. Let us consider the set $\Histories$ of all the correct histories of this system projected to object $R$.
Let $\States$ be the set of states extracted from $\Histories$. Observe that the states in $\States$ only contain events from two processes: the writer $p_w$ and reader $p_r$ processes. Let us assume by way of contradiction that $R$ does not satisfy asynchrony in $\States$, then there is a $\state \in \States$ and events $e_w$ and $e_r$ from writer and reader respectively such that
\begin{equation}
    (\state \cup \{e_w\} \in \States) \land (\state \cup \{e_r\} \in \States) \land  (\state \cup \{e_w, e_r\} \notin \States).
\end{equation}

This implies that $\state$ can be extracted from a history $H_w$ from which $\state_w= \state \cup \{e_w\}$ can also be extracted, but no such history $H_w$ has event $e_r$.
Similarly, $\state$ can be extracted from a history $H_r$ from which $\state_r= \state \cup \{e_r\}$ can also be extracted, but no such history $H_r$ has event $e_w$.
We have that $e_w$ is an event from a write op-ex, and hence $e_w \equiv R.\wwrite.\inv$ or $e_w \equiv R.\wwrite.\res$. On its hand, $e_r$ is an event from a read op-ex, and $e_r \equiv R.\rread.\inv$ or $e_r \equiv R.\rread.\res$. 
We have the following possibilities:

(1) First, consider a situation in which one of the events is an invocation event (\ie $e_r \equiv R.\rread.\inv$ or $e_w \equiv R.\wwrite.\inv$). 
Let us assume, without loss of generality, that $e_w \equiv R.\wwrite.\inv$. We have that $\state|p_w=\state_r|p_w$. Then, from the process consistent behavior assumption (\Cref{asm:ProcessConsistenBehavior}) applied to $\state$, $\state_r$, and $e_w$, we have that $\state_r \cup \{e_w\}=\state \cup \{e_w, e_r\}$ belongs to $\States$, which contradicts the assumption. The case $e_r \equiv R.\rread.\inv$ is similar.

(2) Next, consider the situation where both events are responses, \ie $e_r \equiv R.\rread.\res$ and $e_w \equiv R.\wwrite.\res$.
Consider $\state_r$, which must contains the invocation $e$ of the write op-ex $o=(e,e_w)$.
Let us consider any history $H$ from which $\state_r$ can be extracted in which $p_w$ is correct.
Then by the \Legality of $H$ (and in particular the \Liveness predicate of the $\wwrite$ operation), op-ex $o$ has to terminate in $H$.
That is, $p_w$ will have $e_w$ as its next event in $H$.
Then $\state_r \cup \{e_w\}=\state \cup \{e_w, e_r\} \in \States$ which is a contradiction.
\end{proof}

\subsubsection{Asynchrony of a Point-to-point Message-passing Object}

We can also prove that a message-passing object as defined in \Cref{sec:async-msg-pass} satisfies the \Asynchrony axiom. We consider here a message passing object $M$ used by two processes, a sender $p_s$ and a receiver $p_r$, to communicate.

\begin{theorem}
    \label{thm:Async-message-passing}
    A point-to-point message-passing object $M$ satisfies the asynchronous distributed system axiom of \Cref{def:asynch-short}.
\end{theorem}

The proof is similar to the proof of \Cref{thm:Async-atomicSWSR} replacing the writer with the sender and the reader with the receiver, and is omitted.

\section{Conclusion} \label{sec:conclusion}
In this paper, we have introduced a modular framework for specifying distributed \nn{objects}.
Our approach departs from sequential specifications, and it deploys simple logic for specifying the interface between the system's components as concurrent objects.
It also separates the object's semantics from other aspects such as consistency and failures, while providing a structured precondition/postcondition style for specifying objects.

We demonstrate the usability of our framework by specifying communication media, services, and even problems, as objects.
\jw{With our formalism, we also provide a proof of the impossibility of consensus}
that is agnostic of the medium used for inter-process communication.
The simple specification examples we presented in this paper were for illustration and understanding the formalism.
\jw{Of course, we acknowledge that some combinations of system model, object, and consistency may not be specified with the current version of the framework.}

We are confident that our framework's expressiveness (via the specification and combination of concurrent objects) enables the specification of more complex distributed systems, \nn{including ones with dynamic node participation}.
As our formalism gets used and flourishes with object definitions, its usefulness will be apparent both to distributed computing researchers and practitioners seeking for a modular specification of complex distributed objects.
In addition, we plan to explore how to feed our specification into proof assistants such as Coq~\cite{huet_coq_1997} and Agda~\cite{bove_brief_2009}.

\bibliographystyle{plain}
\bibliography{bibliography}

\begin{thebibliography}{10}

\bibitem{AJS05}
Ali~E. Abdallah, Cliff~B. Jones, and Jeff~W. Sanders, editors.
\newblock {\em Communicating Sequential Processes: The First 25 Years,
  Symposium on the Occasion of 25 Years of CSP}, volume 3525 of {\em Lecture
  Notes in Computer Science}. Springer, 2005.

\bibitem{ANBKH95}
Mustaque Ahamad, Gil Neiger, James~E. Burns, Prince Kohli, and Phillip~W.
  Hutto.
\newblock Causal memory: Definitions, implementation, and programming.
\newblock {\em Distributed Comput.}, 9(1):37--49, 1995.

\bibitem{AFRT23}
Timoth\'e Albouy, Davide Frey, Michel Raynal, and Fran{\c{c}}ois Ta{\"{\i}}ani.
\newblock Asynchronous {B}yzantine reliable broadcast with a message adversary.
\newblock {\em Theor. Comput. Sci.}, 978:114110, 2023.

\bibitem{AW04}
Hagit Attiya and Jennifer~L. Welch.
\newblock {\em Distributed computing - fundamentals, simulations, and advanced
  topics {(2.} ed.)}.
\newblock Wiley series on parallel and distributed computing. Wiley, 2004.

\bibitem{BHG87}
Philip~A. Bernstein, Vassos Hadzilacos, and Nathan Goodman.
\newblock {\em Concurrency Control and Recovery in Database Systems}.
\newblock Addison-Wesley, 1987.

\bibitem{bove_brief_2009}
Ana Bove, Peter Dybjer, and Ulf Norell.
\newblock A {Brief} {Overview} of {Agda} – {A} {Functional} {Language} with
  {Dependent} {Types}.
\newblock In {\em Theorem {Proving} in {Higher} {Order} {Logics}}, pages
  73--78, Berlin, Heidelberg, 2009. Springer.

\bibitem{B87}
Gabriel Bracha.
\newblock Asynchronous {Byzantine} agreement protocols.
\newblock {\em Inf. Comput.}, 75(2):130--143, 1987.

\bibitem{B14}
Sebastian Burckhardt.
\newblock Principles of eventual consistency.
\newblock {\em Found. Trends Program. Lang.}, 1(1-2):1--150, 2014.

\bibitem{BGYZ14}
Sebastian Burckhardt, Alexey Gotsman, Hongseok Yang, and Marek Zawirski.
\newblock Replicated data types: specification, verification, optimality.
\newblock In {\em Proc. 41st {ACM} {SIGPLAN-SIGACT} Symposium on Principles of
  Programming Languages (POPL'14)}, pages 271--284. {ACM}, 2014.

\bibitem{CRR18}
Armando Casta{\~{n}}eda, Sergio Rajsbaum, and Michel Raynal.
\newblock Unifying concurrent objects and distributed tasks:
  Interval-linearizability.
\newblock {\em J. {ACM}}, 65(6):45:1--45:42, 2018.

\bibitem{CRR23}
Armando Casta{\~{n}}eda, Sergio Rajsbaum, and Michel Raynal.
\newblock A linearizability-based hierarchy for concurrent specifications.
\newblock {\em Commun. {ACM}}, 66(1):86--97, 2023.

\bibitem{CK21}
Shir Cohen and Idit Keidar.
\newblock Tame the wild with {Byzantine} linearizability: Reliable broadcast,
  snapshots, and asset transfer.
\newblock In {\em Proc. 35th Int'l Symposium on Distributed Computing
  (DISC'21)}, volume 209 of {\em LIPIcs}, pages 18:1--18:18. Schloss Dagstuhl -
  Leibniz-Zentrum f{\"{u}}r Informatik, 2021.

\bibitem{EH83}
E.~Allen Emerson and Joseph~Y. Halpern.
\newblock {``Sometimes''} and ``not never'' revisited: On branching versus
  linear time.
\newblock In {\em Proc. 10th {ACM} Symposium on Principles of Programming
  Languages (POPL'83)}, pages 127--140. {ACM} Press, 1983.

\bibitem{FLP85}
Michael~J. Fischer, Nancy~A. Lynch, and Mike Paterson.
\newblock Impossibility of distributed consensus with one faulty process.
\newblock {\em J. {ACM}}, 32(2):374--382, 1985.

\bibitem{GL23}
Eli Gafni and Giuliano Losa.
\newblock Invited paper: Time is not a healer, but it sure makes hindsight
  20:20.
\newblock In {\em Proc. 25th Int'l Symposium on Stabilization, Safety, and
  Security of Distributed Systems (SSS'23)}, volume 14310 of {\em Lecture Notes
  in Computer Science}, pages 62--74. Springer, 2023.

\bibitem{HT94}
Vassos Hadzilacos and Sam Toueg.
\newblock A modular approach to fault-tolerant broadcasts and related problems.
\newblock Technical report, Cornell University, 1994.

\bibitem{HRT98}
Maurice Herlihy, Sergio Rajsbaum, and Mark~R. Tuttle.
\newblock Unifying synchronous and asynchronous message-passing models.
\newblock In {\em Proc 17th {ACM} Symposium on Principles of Distributed
  Computing (PODC'98)}, pages 133--142. {ACM}, 1998.

\bibitem{HS99}
Maurice Herlihy and Nir Shavit.
\newblock The topological structure of asynchronous computability.
\newblock {\em J. {ACM}}, 46(6):858--923, 1999.

\bibitem{HW90}
Maurice Herlihy and Jeannette~M. Wing.
\newblock Linearizability: A correctness condition for concurrent objects.
\newblock {\em {ACM} Trans. Program. Lang. Syst.}, 12(3):463--492, 1990.

\bibitem{HS97}
Gunnar Hoest and Nir Shavit.
\newblock Towards a topological characterization of asynchronous complexity
  (preliminary version).
\newblock In {\em Proc. 16th {ACM} Symposium on Principles of Distributed
  Computing (PODC'97)}, pages 199--208. {ACM}, 1997.

\bibitem{huet_coq_1997}
Gérard Huet, Gilles Kahn, and Christine Paulin-Mohring.
\newblock The {Coq} proof assistant a tutorial.
\newblock {\em Rapport Technique}, 178, 1997.

\bibitem{L78}
Leslie Lamport.
\newblock Time, clocks, and the ordering of events in a distributed system.
\newblock {\em Commun. {ACM}}, 21(7):558--565, 1978.

\bibitem{L79}
Leslie Lamport.
\newblock How to make a multiprocessor computer that correctly executes
  multiprocess programs.
\newblock {\em {IEEE} Trans. Computers}, 28(9):690--691, 1979.

\bibitem{L94}
Leslie Lamport.
\newblock The temporal logic of actions.
\newblock {\em {ACM} Trans. Program. Lang. Syst.}, 16(3):872--923, 1994.

\bibitem{LS88}
Richard~J. Lipton and Jonathan~S. Sandberg.
\newblock {PRAM}: A scalable shared memory.
\newblock Technical Report TR-180-88, Princeton University, 1988.

\bibitem{L96}
Nancy~A. Lynch.
\newblock {\em Distributed Algorithms}.
\newblock Morgan Kaufmann, 1996.

\bibitem{LT87}
Nancy~A. Lynch and Mark~R. Tuttle.
\newblock Hierarchical correctness proofs for distributed algorithms.
\newblock In {\em Proc. 16th {ACM} Symposium on Principles of Distributed
  Computing (PODC'87)}, pages 137--151. {ACM}, 1987.

\bibitem{N94}
Gil Neiger.
\newblock Set-linearizability.
\newblock In {\em Proc. 13th {ACM} Symposium on Principles of Distributed
  Computing (PODC'94)}, page 396. {ACM}, 1994.

\bibitem{P16}
Matthieu Perrin.
\newblock {\em Sp{\'{e}}cification des objets partag{\'{e}}s dans les
  syst{\`{e}}mes r{\'{e}}partis sans-attente. (Specification of shared objects
  in wait-free distributed systems)}.
\newblock PhD thesis, University of Nantes, France, 2016.

\bibitem{PMJ16}
Matthieu Perrin, Achour Most{\'{e}}faoui, and Claude Jard.
\newblock Causal consistency: beyond memory.
\newblock In {\em Proc. 21st {ACM} {SIGPLAN} Symposium on Principles and
  Practice of Parallel Programming (PPoPP'16)}, pages 26:1--26:12. {ACM}, 2016.

\bibitem{P77}
Amir Pnueli.
\newblock The temporal logic of programs.
\newblock In {\em Proc. 18th Symposium on Foundations of Computer Science
  (FOCS'77)}, pages 46--57. {IEEE} Computer Society, 1977.

\bibitem{R18}
Michel Raynal.
\newblock {\em Fault-Tolerant Message-Passing Distributed Systems - An
  Algorithmic Approach}.
\newblock Springer, 2018.

\bibitem{SW89}
Nicola Santoro and Peter Widmayer.
\newblock Time is not a healer.
\newblock In {\em Proc. 6th Symposium on Theoretical Aspects of Computer
  Science (STACS'89)}, volume 349 of {\em Lecture Notes in Computer Science},
  pages 304--313. Springer, 1989.

\bibitem{S11}
Nir Shavit.
\newblock Data structures in the multicore age.
\newblock {\em Commun. {ACM}}, 54(3):76--84, 2011.

\bibitem{SN04}
Robert~C. Steinke and Gary~J. Nutt.
\newblock A unified theory of shared memory consistency.
\newblock {\em J. {ACM}}, 51(5):800--849, 2004.

\bibitem{T91}
Gadi Taubenfeld.
\newblock On the nonexistence of resilient consensus protocols.
\newblock {\em Inf. Process. Lett.}, 37(5):285--289, 1991.

\bibitem{VV16}
Paolo Viotti and Marko Vukolic.
\newblock Consistency in non-transactional distributed storage systems.
\newblock {\em {ACM} Comput. Surv.}, 49(1):19:1--19:34, 2016.

\bibitem{VF03}
Roman Vitenberg and Roy Friedman.
\newblock On the locality of consistency conditions.
\newblock In {\em Proc. 17th Int'l Conference on Distributed Computing
  (DISC'03)}, volume 2848 of {\em Lecture Notes in Computer Science}, pages
  92--105. Springer, 2003.

\bibitem{V04}
Hagen V{\"{o}}lzer.
\newblock A constructive proof for {FLP}.
\newblock {\em Inf. Process. Lett.}, 92(2):83--87, 2004.

\end{thebibliography}

 \newpage
\appendix

\section*{\LARGE Appendix}

\section{FIFO Consistency Addendum} \label{sec:fifo-addendum}

Let us notice that the definition of the \FIFOConsistency condition, as defined in \Cref{sec:cons-cond}, differs from the traditional definition of PRAM consistency we encounter in the literature, which is: ``For each process $p_i$, we can construct a total order of op-exes containing op-exes of $p_i$, and the \textit{update} op-exes of all processes.''
Here, \textit{update} op-exes refer to the op-exes that change the object's internal state.
For instance, for a register object with \rread and \wwrite operations, the updates would be the \wwrite op-exes.
However, this initial definition is not completely accurate, because when we are constructing the total order of op-exes for a process $p_i$, if some update op-ex of another process $p_j$ returns a value, we do not want to verify the validity of this value.
Furthermore, adding a ``$\Forall p_i \in \Processes$'' quantifier at the start of the \FIFOConsistency condition would make this condition structurally different from the other conditions of \Cref{def:classic-consistency}, as it would create a potentially different \oo relation for every process of the system, instead of having a single global \oo relation like the other conditions of \Cref{def:classic-consistency}.

Hence, our definition of \FIFOConsistency relies on our new predicate $\OpFIFOOrder(H,\oo)$, which enforces a specific pattern on the \oo relation that characterizes the FIFO order of op-exes.
As a reminder, here are the definitions of \OpFIFOOrder and \FIFOConsistency given in \Cref{sec:opex-orders,sec:cons-cond}, respectively.
\begin{align*}
    \OpFIFOOrder(H,\oo)\! &\triangleq \Forall o_i,o_i' \in O|p_i, o_j,o_j' \in O|p_j: \\
    &\hspace{3.2em}(o_i \oo o_i' \oo o_j \oo o_j' \land o_i \oo o_j') \implies (o_i \oo o_j \land o_i' \oo o_j'). \\
    \FIFOConsistency(H,\oo) &\triangleq \ProcessConsistency(H,\oo) \cup \{\OpFIFOOrder(H,\oo)\}.
\end{align*}

\begin{figure}[ht]
    \centering
    \begin{tikzpicture}
    
\tikzmath{
\yOffset=1.75;
\xLength=7;
\circleR=.6;
}

\node (pi) at (0,\yOffset) {$p_i$};
\node (pj) at (0,0) {$p_j$};

\node[circle,draw,minimum size=\circleR cm] (oi1) at (\xLength/2-2,\yOffset) {};
\node[circle,draw,minimum size=\circleR cm] (oi2) at (\xLength/2-.5,\yOffset) {};
\node[circle,draw,minimum size=\circleR cm] (oj1) at (\xLength/2+.5,0) {};
\node[circle,draw,minimum size=\circleR cm] (oj2) at (\xLength/2+2,0) {};

\draw[->,dashed] (pi) -- (oi1) (oi2) -- (\xLength,\yOffset);
\draw[->,dashed] (pj) -- (oj1) (oj2) -- (\xLength,0);

\node at (\xLength/2-2,\yOffset) {$o_i$};
\node at (\xLength/2-.5,\yOffset) {$o_i'$};
\node at (\xLength/2+.5,0) {$o_j$};
\node at (\xLength/2+2,0) {$o_j'$};

\draw[->] (oi1) -- (oi2);
\draw[->] (oj1) -- (oj2);

\draw[->] (oi1) -- (oj2);
\draw[->] (oi2) -- (oj1);


\draw[->,dotted,thick] (oi1) -- (oj1);
\draw[->,dotted,thick] (oi2) -- (oj2);
    
\end{tikzpicture}
    \caption{Illustration of \OpFIFOOrder: if the pattern represented by the 4 hard arrows is present in the \oo op-ex relation, then the 2 dotted arrows must also be present in \oo.}
    \label{fig:fifo-order}
\end{figure}
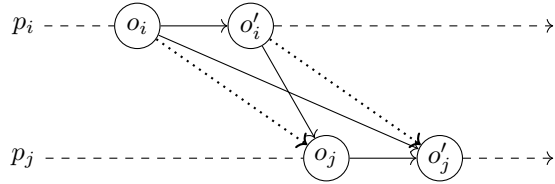

As said in \Cref{sec:opex-orders}, intuitively, \OpFIFOOrder checks that a given op-ex sees all its predecessors on the same process, plus all the predecessors of the op-exes it sees on other processes.
Furthermore, the ``knowledge'' of the op-exes of a given process is monotonically increasing with time: all the op-exes seen by a given op-ex must also be seen by its successors on the same process. 
\Cref{fig:fifo-order} illustrates the \OpFIFOOrder predicate: we consider two processes, $p_i$ and $p_j$ (that can be the same process), that both have two op-exes ($o_i$, $o_i'$, $o_j$ and $o_j'$).
If the first op-ex of $p_j$ sees the second op-ex of $p_i$, and the second op-ex of $p_j$ sees the first op-ex of $p_i$, then the first and second op-ex of $p_j$ must respectively see the first and second op-ex of $p_i$. 

In the end, we believe that, with this \OpFIFOOrder predicate, the resulting definition of \FIFOConsistency that we obtain is simpler than the original definition of PRAM consistency, while still achieving the same goal.











\end{document}